\documentclass[11pt]{article} 

\usepackage[utf8]{inputenc} % allow utf-8 input
\usepackage[T1]{fontenc}    % use 8-bit T1 fonts
\usepackage{hyperref}       % hyperlinks
\usepackage{url}            % simple URL typesetting
\usepackage{booktabs}       % professional-quality tables
\usepackage{amsfonts}       % blackboard math symbols
\usepackage{nicefrac}       % compact symbols for 1/2, etc.
\usepackage{microtype}      % microtypography

\usepackage[square,sort&compress,comma,numbers]{natbib}
\usepackage{tcolorbox}
\usepackage{graphicx}
\usepackage{enumitem}
\usepackage{array}
\usepackage{amsmath}
\usepackage{amssymb}
\usepackage{amsthm}
\usepackage{mathabx}
\usepackage{subcaption}
\usepackage[linesnumbered,ruled,vlined ]{algorithm2e}
\usepackage{makecell}

\renewcommand{\a}{\mathbf{a}} 
 
\newcommand{\x}{\mathbf{x}}

\newcommand{\y}{\mathbf{y}}
\newcommand{\w}{\mathbf{w}}
\newcommand{\bw}{\mathbf{\bar w}}
\renewcommand{\r}{\mathbf{r}}

\newcommand{\p}{\mathbf{p}} 
\newcommand{\bp}{\mathbf{\bar p}} 
\newcommand{\I}{\mathbf{I}}
\newcommand{\X}{\mathbf{X}}

\newcommand{\R}{\mathbb{R}}
\renewcommand{\H}{\mathbf{H}}
\renewcommand{\P}{\mathbb{P}}
\newcommand{\g}{\mathbf{g}}

\newcommand{\nn}{\nonumber}
\newcommand{\as}{\alpha^\star}
\newcommand{\f}{\ensuremath{f}}
\renewcommand{\ln}{LocalNewton}
\renewcommand{\S}{\mathcal{S}}

\newcommand*\bS{\mathbf{S}}
\newcommand*\bB{\mathbf{B}}

\newcommand*\mb{\mathbf{b}}

\newcommand{\su}{\sum_{k=1}^K}
\newcommand{\sut}{\sum_{\tau=t_0}^t}

\newtheorem{theorem}{Theorem}[section]

\newtheorem{lemma}[theorem]{Lemma}

\newtheorem{remark}{Remark}

\usepackage{authblk}
\title{LocalNewton: Reducing Communication Bottleneck for Distributed Learning}

\author[1]{Vipul Gupta\thanks{Equal Contribution}}
\author[1]{Avishek Ghosh$^*$}
\author[2]{Micha{\l} Derezi\'nski}
\author[2]{Rajiv Khanna}
\author[1]{Kannan Ramchandran}
\author[2]{Michael W. Mahoney}

\affil[1]{Department of EECS, University of California, Berkeley}
\affil[2]{Department of Statistics, University of California, Berkeley}
\date{\vspace{-5ex}}

\begin{document}

\maketitle

\begin{abstract}

To address the communication bottleneck problem in distributed optimization within a master-worker framework, we propose LocalNewton, a distributed second-order algorithm with \emph{local averaging}. 
In LocalNewton, the worker machines update their model in every iteration by finding a suitable second-order descent direction using only the data and model stored in their own local memory. 
We let the workers run multiple such iterations locally and communicate the models to the master node only once every few (say $L$) iterations. 
LocalNewton is highly practical since it requires only one hyperparameter, the number $L$ of local iterations. 
We use novel matrix concentration based techniques to obtain theoretical guarantees for LocalNewton, and we validate them with detailed empirical evaluation. 
To enhance practicability, we devise an adaptive scheme to choose $L$, and we show that this reduces the number of local iterations in worker machines between two model synchronizations as the training proceeds, successively refining the model quality at the master. 
Via extensive experiments using several real-world datasets with AWS Lambda workers and an AWS EC2 master, we show that LocalNewton requires fewer than $60\%$ of the communication rounds (between master and workers) and less than $40\%$ of the end-to-end running time, compared to state-of-the-art algorithms, to reach the same training~loss. 
\end{abstract}

\section{Introduction}
\label{sec:intro}

An explosion in data generation and data collection capabilities in recent years has resulted in the segregation of computing and storage resources. 
Distributed machine learning is one example where each worker machine processes only a subset of the data, while the master machine coordinates with workers to learn a good model. 
Such coordination can be time-consuming since it requires frequent communication between the master and worker nodes, especially for systems that have large compute resources, but are bottlenecked by communication costs. 

Communication costs in distributed optimization can be broadly classified into two types---(a) latency cost and (b) bandwidth cost \citep{demmel2,oversketch}. 
Latency is the fixed cost associated with sending a message, and it is generally independent of the size of the message. 
Bandwidth cost, on the other hand, is directly proportional to the size of the message. 
Many recent works have focused on reducing the bandwidth cost by reducing the size of the gradient or the model to be communicated using techniques such as sparsification \citep{sparse_comm_distributed2019,stich2018sparsified}, sketching \citep{sgd_with_sketching,konevcny2016federated} and quantization \citep{ghosh_one,lin2017deep,bernstein2018signsgd, hawq, qbert}. 
Schemes that perform inexact updates in each iteration, however, can \emph{increase} the number of iterations required to converge to the same quality model, both theoretically and empirically \citep{sparse_comm_distributed2019,tyagi_limits_quantized2020}. 
This can, in turn, increase the total training time in systems 
% (e.g., federated learning and serverless computing) 
where latency costs dominate bandwidth costs. 
\begin{figure*}[t]
\centering
    \begin{subfigure}{.48\textwidth}
        \centering
        \includegraphics[scale=0.6]{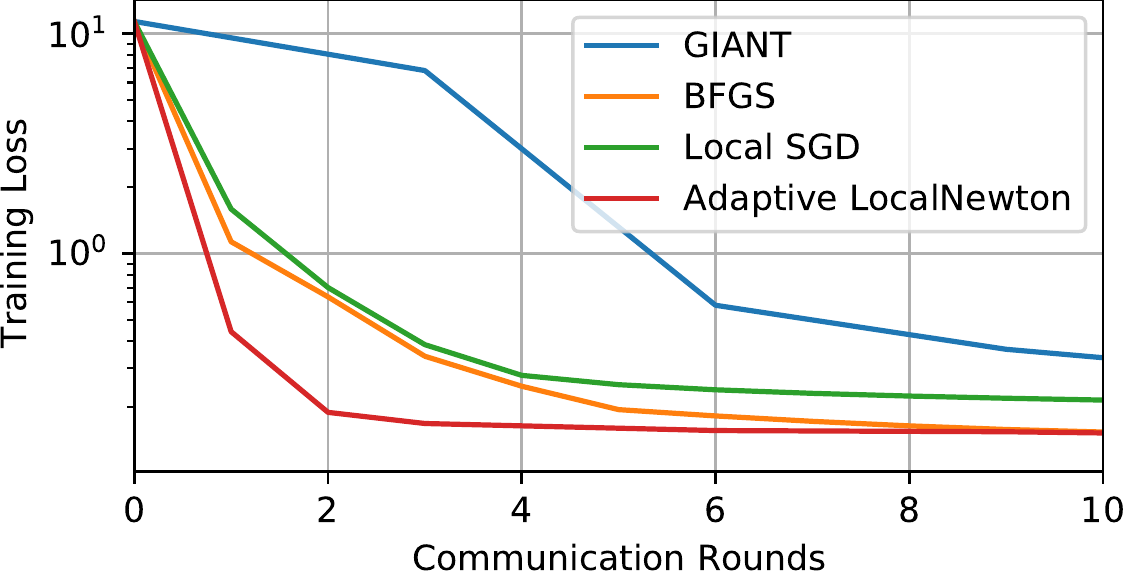}
    \end{subfigure}
    ~~
    \begin{subfigure}{.48\textwidth}
      \centering
    \includegraphics[scale=0.6]{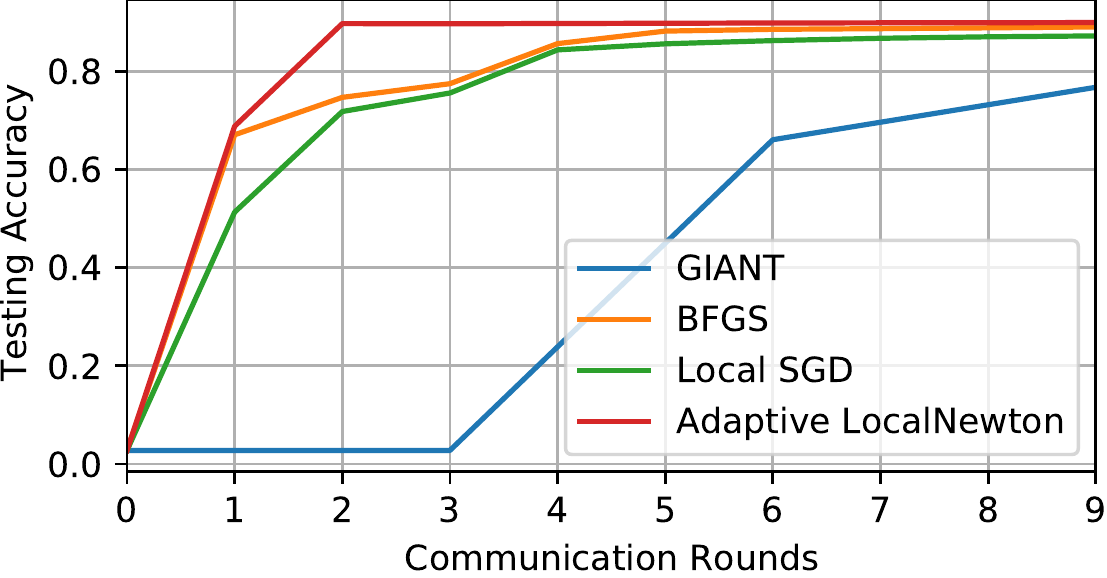}
    % \caption{Covtype dataset}
    \end{subfigure}
    \caption{Training loss and testing accuracy 
    versus communication rounds
    for Adaptive \ln~and existing schemes for communication-efficient optimization. }
    \label{fig:illus}
\end{figure*}

One setting where latency costs may outweigh bandwidth costs is federated learning, where the computation is performed locally at the mobile device (which is generally the source of the data) due to a high-cost barrier in transferring the data to traditional computing platforms \citep{federated_learning_survey2019,konevcny2016federated}. 
Such mobile resources (e.g., mobile phones, wearable devices, etc.) have reasonable compute power, but they can be severely limited by communication latency (e.g., inadequate network connection to synchronize frequent model updates). 
For this reason, schemes like Local Stochastic Gradient Descent (Local SGD) have become popular, since they try to mitigate the communication costs by performing more \emph{local} computation at the worker machines, thus substantially reducing the number of communication rounds required~\citep{pmlr-v54-mcmahan17a}. 

Serverless systems---such as Amazon Web Services (AWS) Lambda and Microsoft Azure Functions---are yet another example where high communication latency between worker machines dominates the running time of the algorithm \citep{osn, hellerstein}. 
% Serverless computing has recently gained a lot of attention from the research community \citep{serverless_computing,berkeley_view,pywren,numpywren}, with a significant focus on optimization in these settings \citep{training_DNN_serverless,osn,cirrus,serverless_ml}. 
Such systems are gaining popularity due to the ease-of-management, greater elasticity and high scalability \citep{serverless_computing, berkeley_view}.
These systems use cloud storage (e.g., AWS S3) to store enormous amounts of data, while using a large number of low-quality workers for large-scale computation. 
Naturally, the communication between the high-latency storage and the commodity workers is extremely slow \citep[e.g., see][]{pywren}, resulting in impractical end-to-end times for many popular optimization algorithms such as SGD \citep{hellerstein,osn}. Furthermore, communication failures between the cloud storage and serverless workers consistently give rise to stragglers, and this introduces synchronization delays~\citep{local_codes, mert_polar_serverless}. 

These trends suggest that optimization schemes that reduce communication rounds between workers are highly desirable. In this paper, we go one-step forward---we propose and analyze a second order method with \emph{local} computations. Being a second order algorithm, the iteration complexity of LocalNewton is inherently low. Moreover, its local nature further cuts down the communication cost between the worker and the master node. To the best of our knowledge, this is the first work to propose and analyze a second order optimization algorithm with local averaging. 
% \michael{Ugh.  This is one place we need to redo text to clarify optimization versus initialization issue.}

Additionally, in this paper, we introduce an adaptive variant of the \ln~algorithm, namely Adaptive \ln. This algorithm chooses the number of local iterations adaptively at the master after each communication round by observing the change in training loss. Thus, it further refines the iterates obtained through \ln~by adaptively and successively reducing the number of local iterations, thereby improving the quality of the model updates. Furthermore, when the number of local iterations, $L$, reduces to $1$, Adaptive LocalNewton automatically switches to a standard second order optimization algorithm, namely GIANT, proposed in ~\citep{giant_nips}.

In Sec.~\ref{sec:theoretical_guarantees}, we show that the iterates of \ln~converge to a norm ball (with small radius) around the global minima. From this point of view, we may think of \ln~as an initialization algorithm, rather than an optimization one; since it takes the iterates close to the optimal point in only a few rounds of communication. After reaching sufficiently close to the solution point, one may choose some standard optimization algorithm to reach to the solution. In Adaptive \ln, we first exploit the fast convergence of \ln~($L \geq 1$), and afterwards, when $L=1$, Adaptive \ln~switches to the standard optimization algorithm (e.g., GIANT \citep{giant_nips}).

{\bf Our contributions.} 
Inspired by recent progress in local optimization methods (that reduce communication cost by limiting the frequency of synchronization) and distributed and stochastic second order methods (that use the local curvature information), we propose a local second-order algorithm called \ln.
The proposed \ln\ method saves on communication costs in two ways. 
First, it updates the models at the master only sporadically, thus requiring only one communication round per multiple iterations. 
Second, it uses the second-order information to reduce the number of iterations, and hence it reduces the overall rounds of communication. 

Important features of \ln~include:

1. \emph{Simplicity}: In \ln, each worker takes only a few Newton steps \citep{boyd} on local data, agnostic of other workers. These local models are then averaged once every $L (\geq 1)$ iterations at the master node.

2. \emph{Practicality}: Unlike many first-order and distributed second-order schemes, \ln~does not require hyperparameter tuning for step-size, mini-batch size, etc., and the only hyperparameter required is the number of local iterations $L$. 
%We also propose a method to adaptively choose $L$ in Sec. \ref{sec:adapt_local_newton}. 
We also propose Adaptive \ln, an adaptive version which automatically reduces $L$ as the training proceeds by monitoring the training loss at the master.
% Through empirical evaluations on multiple datasets, we show that adaptive \ln~is robust to the exact value of $L$.

3. \emph{Convergence guarantees}: 
In general, proving convergence guarantees for local algorithms is not straightforward.  Only recently, it has been proved \citep{stich2018local} that local SGD converges as fast as SGD, thereby explaining the well-studied empirical successes~\citep{konevcny2016federated}. 
In this paper, we develop novel techniques to highlight the  convergence behaviour of \ln.

4. \emph{Reduced training times}: We implement \ln~on the Pywren framework \citep{pywren} using AWS Lambda\footnote{A high-latency \emph{serverless} computing platform.} workers and an AWS EC2\footnote{A traditional \emph{serverful} computing platform.} master. Through extensive empirical evaluation, we show that the significant savings in terms of communication rounds translate to savings in running time on this high-latency distributed computing environment.

5. \emph{Adaptivity}: We propose Adaptive \ln, which is an adaptive variant of the \ln~algorithm. In the adaptive scheme, based on the change in function objective value, the master modulates the number of local iterations at the worker machines. This improves the quality of the model updates as discussed in detail in Sec. \ref{sec:adapt_local_newton}.

6. \emph{LocalNewton as an initialization algorithm}: Since the convergence guarantees of LocalNewton (see Sec.~\ref{sec:theoretical_guarantees}) only ensure that the iterates stay in a norm ball around the minima,  one may rethink \ln~as an initialization algorithm, rather than an optimization one. In only a very few communication rounds, LocalNewton takes the iterates very close to the optimal solution. After that, our algorithm switches to a standard second-order algorithm, GIANT ~\citep{giant_nips}.

Fig. \ref{fig:illus} illustrates savings due to Adaptive \ln, where we plot training loss and test accuracy with communication rounds, for several popular communication-efficient schemes for logistic regression on the w8a dataset \citep{libsvm} (see Sec. \ref{sec:experiments} for a details on experiments). 
Observe that Adaptive \ln~reaches close to the optimal training loss very quickly, when compared to schemes like Local SGD \citep{konevcny2016federated, stich2018local}, GIANT \citep{giant_nips} and BFGS \citep{bfgs_fletcher2013practical}.

%%\subsection{Related Work and Our Contributions}
{\bf Related Work.}
%%% {\bf Local first-order methods}: 
In recent years, schemes such as local SGD have gained popularity, as they are communication-efficient due to only sporadic model updates at the master \citep{konevcny2016federated, pmlr-v54-mcmahan17a, federated_learning_survey2019, MLSYS2019_c8ffe9a5}. 
Such schemes that show great promise have eluded a thorough theoretical analysis until recently, when it was shown that local SGD converges at the same rate as mini-batch SGD \citep{stich2018local,haddadpour2019local,patel_aymeric_nips2019}. Similar ideas that reduce communication by averaging the local models sporadically have also been applied in training neural networks to improve the training times and/or model performance \citep{local_sgd_dnn,swap}.
Apart from local first-order methods, 
%%% like local SGD that reduce communication by updating the model only every tens of iterations, 
many distributed second-order (also known as Newton-type) algorithms have been proposed to reduce communication (e.g., \cite{jadbabaie2009distributed} and \cite{tutunov2019distributed} propose second-order algorithms in a decentralized setting where communication happens over a predefined graph, and \cite{derezinski2018batch} proposes batch-expansion training which reduces communication in early phases of the optimization).
More recently, many communication-efficient second-order methods were proposed in the centralized setting, e.g., see
\citep{giant_nips,dane,disco,cocoa, aide, pmlr-v80-duenner18a,det_avg,derezinski2020debiasing,ghosh_comrade,ghosh2020communication}. 
Such methods use both the gradient and the curvature information to provide an order of improvement in convergence, compared to vanilla first-order methods. 
This is done at the cost of more local computation per iteration, which is acceptable for systems with high communication latency. 
However, such algorithms require at least two communication rounds (for averaging gradients and the second-order descent direction), and a thorough knowledge of a fundamental trade-off between communication and local computation is still lacking for these methods. % In the present work, we fill this gap.
Stochastic second order optimization theory has been developed recently~\citep{fred1,fred2,Fred:non-convex}, and second order implementations motivated by this theory have been shown to outperform state-of-the-art~\citep{yao2019pyhessian,YGSKM20_adahessian_TR}.

\section{Problem Formulation}

We first define the notation used and then introduce the basic problem setup considered in this paper.

{\bf Notation.}  Throughout the paper, vectors (e.g., $\g$) and matrices (e.g., $\H$) are represented as bold lowercase and uppercase letters, respectively. For a vector $\g$, $\|\g\|$ denotes its $\ell_2$ norm, and $\|\H\|_2$ denotes the spectral norm of matrix $\H$. The identity matrix is denoted as $\I$,  and the set $\{1,2,\cdots, n\}$ is denoted as $[n]$, for all positive integers $n$.
Further, we use superscript (e.g., $\g^k$) to denote the worker index and subscript (e.g., $\g_t$) to denote the iteration counter (i.e., time index), unless stated otherwise.

% \subsection{Problem Formulation}
{\bf Problem Setup.} We are interested in solving empirical risk minimization problems of the following form in a distributed fashion
\begin{equation}\label{mainProblem}
    \min _{\mathbf{w} \in \mathbb{R}^{d}}\left\{f(\mathbf{w}) \triangleq \frac{1}{n} \sum_{j=1}^{n} f_{j}\left(\mathbf{w}\right)\right\},
\end{equation}
where $f_j(\cdot): \R^d\rightarrow \R$, for all $j\in [n] = \{1,2,\cdots,n\}$, models the loss of the $j$-th observation given an underlying parameter estimate $\w\in \R^d$. In machine learning, such problems arise frequently, e.g., logistic and linear regression, support vector machines, neural networks and graphical models. 
Specifically, in the case of logistic regression, 
$$f_j(\w) = \ell_j(\w^T\x_j) = \log(1 + e^{-y_j\w^T\x_j}) + \frac{\gamma}{2} \|\w\|^2,$$
where $\ell_j(\cdot)$ is the loss function for sample $j\in[n]$ and $\gamma$ is an appropriately chosen regularization parameter. 
Also, $\X = [\x_1, \x_2, \cdots, \x_n] \in \R^{d\times n}$ is the sample matrix containing the input feature vectors $\x_j \in \R^d, j \in [n]$, and $\y = [y_1, y_2, \cdots, y_n]$ is the corresponding label vector. 
Hence, $(\x_j, y_j)$ together define the $j$-th observation and $(\X,\y)$ define the training~dataset.

For such problems,
the gradient and the Hessian at the $t$-th iteration are given by
\begin{align*}
\mathbf{g}_{t}&=\nabla f\left(\mathbf{w}_{t}\right)=\frac{1}{n} \sum_{j=1}^{n} \nabla f_j(\w_t) \text{ and } \\
\mathbf{H}_{t}&=\nabla^{2} f\left(\mathbf{w}_{t}\right)=\frac{1}{n} \sum_{j=1}^{n} \nabla^{2} f_j\left(\mathbf{w}_{t}\right),
\end{align*}
respectively, where $\w_t$ is the model estimate at the $t$-th iteration.

{\bf Data distribution at each worker}: Let there be a total of $K$ workers. We assume that the $k$-th worker is assigned a subset $\mathcal{S}_k \subset [n]$, for all $k\in [K] = \{1,2,\cdots, K\}$, of the $n$ data points, chosen uniformly at random without replacement.\footnote{This corresponds simply to partitioning the dataset and assigning an equal number of observations to each worker, if the observations are independent and identically distributed. If not, randomly shuffling the observations and then performing a data-independent partitioning is equivalent to uniform sampling without replacement.} Let the number of samples at each worker be $s = |\mathcal{S}_k| ~\forall~k\in [K]$, where $s \ll n$ in practice. Also, by the virtue of sampling without replacement, we have $\S_1 \cup \S_2 \cup \cdots \cup \S_K = [n]$ and $\S_i \cap \S_j = \Phi$ for all $i,j \in [K]$.
Hence, the number of workers is given by $K = n/s$. 

%Next, we describe \ln, a communication-efficient algorithm for distributed optimization.

\section{Algorithms}
\label{sec:algos}
In this section, we propose two novel algorithms for distributed optimization. First, we propose \ln, a second order algorithm with local averaging. Subsequently, we also propose an adaptive variant of \ln. Here \ln~acts as a good initialization scheme that pushes its iterates close to the optimal solution in a small number of communication rounds. Finally, a standard second-order algorithm is used to converge to the optimal solution.

\subsection{\ln}

We consider synchronous second-order methods for distributed learning, where local models are synced after every $L$ iterations. Let $\mathcal{I}_t \subseteq [t]$ be the set of indices where the model is synced, that is, $\mathcal{I}_t = [0,L,2L, \cdots, t_0]$, where $t_0$ is the last iteration just before $t$ where the models were synced. 

At the $k$-th worker in the $t$-th iteration, the local function value (at the local iterate $\w_t^k$) is 
\begin{equation}\label{local_fn_value}
    f^k(\w_t^k) = \frac{1}{s}\sum_{j\in\S_k}f_{j}(\w_t^k).
\end{equation}

\begin{algorithm}[t]
\SetAlgoLined
\SetKwInOut{Input}{Input}
\Input{Local function $f^k(\cdot)$ at the $k$-th worker; Initial iterate $\bw_0\in\mathbb{R}^d$; Line search parameter $0 < \beta \leq 1/2$; Number of iterations $T$, Set $\mathcal{I}_T \subseteq \{1,2,\cdots,T\}$ where models are synced}
\vspace{5pt}
%\KwResult{$\w^*$}
\For{$k=1$ to $K$ in parallel}{
{\bf Initialization}: $\w_0^k = \bw_0$\\
  \For{$t=0$ to $T-1$}{
  \If{$t\in \mathcal{I}_T$}{
% The current models at each worker are averaged at the master (requiring one round of communication): 
\tcp{Master averages the local models}
$\bw_t = \frac{1}{K}\sum_{k=1}^K \w_t^k$ \\
% Master ships the averaged model to all workers\\
% Worker updates local model
$\w_t^k = \bw_t$
}
% Worker computes the local gradient 
\tcp{Compute the local gradient}
$\g^k(\w_t^k) = \nabla f^k(\w_t^k)$. \\
% Worker computes the local second-order direction
\tcp{Compute the update direction}
$\p_t^k = \H^k(\w_t^k)^{-1}\g^k(\w_t^k)$ \\
Find step-size $\alpha_t^k$ using line search (Eq. \eqref{local-ss})\\
Update model: $\w_{t+1}^k = \w_t^k - \alpha_t^k \p_t^k \quad$ 
}}
 \caption{\ln}
 \label{algo:localnewton}
\end{algorithm}

The $k$-th worker tries to minimize the local function value in Eq \eqref{local_fn_value} in each iteration.
The corresponding local gradient $\g_t^k$ and  local Hessian $\H_t^k$, respectively, at $k$-th worker in $t$-th iteration can be written as 
\begin{align*}
\mathbf{g}_{t}^k = \nabla f^k(\w_t^k) = \frac{1}{s} \sum_{j\in\S_k} \nabla f_j(\w_t^k) \text{ and } \\
\mathbf{H}_{t}^k = \nabla^2 f^k(\w_t^k) = \frac{1}{s} \sum_{j\in\S_k} \nabla^{2} f_j(\mathbf{w}_{t}^k).
\end{align*}

Let us consider the following \ln~update at the $k$-th worker and $(t+1)$-st iteration:
\begin{align}
\label{w_update} 
\w_{t+1}^k = 
&\begin{cases}
\w_t^k - \alpha_t^k\H^k(\w_t^k)^{-1}\g^k(\w_t^k), & \text{ if } t \notin \mathcal{I}_t \\
\bw_t - \alpha_t^k\H^k(\bw_t)^{-1}\g^k(\bw_t), & \text{ if } t \in \mathcal{I}_t ,
\end{cases}
\end{align}
where 
$\bw_t = \frac{1}{K}\sum_{k=1}^K \w_t^k~ \forall ~t,$
and $\alpha_t^k$ is the step-size at the $k$-th worker at iteration $t$.%
\footnote{In practice, one need not calculate the exact $\H^k(\w_t^k)^{-1}\g^k(\w_t^k)$, and efficient algorithms like conjugate gradient descent can be used \citep{cg}.}  

Also, define the local descent direction at the $k$-th worker at iteration $t$ as
$\p_t^k = \alpha_t^k\H^k(\w_t^k)^{-1}\g^k(\w_t^k)$ and similarly define $\bp_t = \frac{1}{K}\sum_{k=1}^K \p_t^k.$
We can see that $\bw_{t+1} = \bw_t - \bp_t$. Detailed steps for \ln~are provided in Algorithm \ref{algo:localnewton}.

Note that $\bw_t$ is not explicitly calculated for all $t$, but only for $t\in \mathcal{I}_t$. However, we will use the technique of perturbed iterate analysis and show the convergence of the sequence $f(\bw_1), f(\bw_2), \cdots, f(\bw_t)$ to $f(\w^*)$. In the next section, we present Adaptive \ln, an algorithm to adaptively choose the number of local iterations, $L$.

\subsection{Adaptive \ln}\label{sec:adapt_local_newton}

\begin{algorithm}[t]
\SetAlgoLined
\SetKwInOut{Input}{Input}
\Input{Minimum decrement $\delta (> 0)$ in global loss function;}
\vspace{5pt}
%\KwResult{$\w^*$}
{\bf Initialization}: $f_{prev} = f(\bw_0)$, Number of local iterations $L$, Set $\mathcal{I}_T \subseteq \{1,2,\cdots,T\}$ where models are synced every $L$ local iterations\\
  \For{$t=1$ to $T-1$}{
  \tcp{$k$-th worker runs \ln~ locally to get $\w_t^k$}
  Workers run Algorithm \ref{algo:localnewton} \\
  \If{$t\in \mathcal{I}_T$}{
% The current models at each worker are averaged at the master (requiring one round of communication): 
\tcp{Master averages the local models}
$\bw_t = \frac{1}{K}\sum_{k=1}^K \w_t^k$ \\
\If{$f_{prev} - f(\bw_t) < \delta$ and $L\geq 1$}{
\tcp{The global function did not decrease enough}
\If{$L=1$}{Switch to GIANT \citep{giant_nips}}
\Else{
Decrease $L$: $L = L-1$\\
Update $\mathcal{I}_T$ according to the new value of $L$
}
}
$f_{prev} = f(\bw_t)$
}}
 \caption{Adaptive \ln}
 \label{algo:adap_localnewton}
\end{algorithm}

\begin{figure*}[t]
    \begin{subfigure}{.45\textwidth}
        \centering
        \includegraphics[scale=0.4]{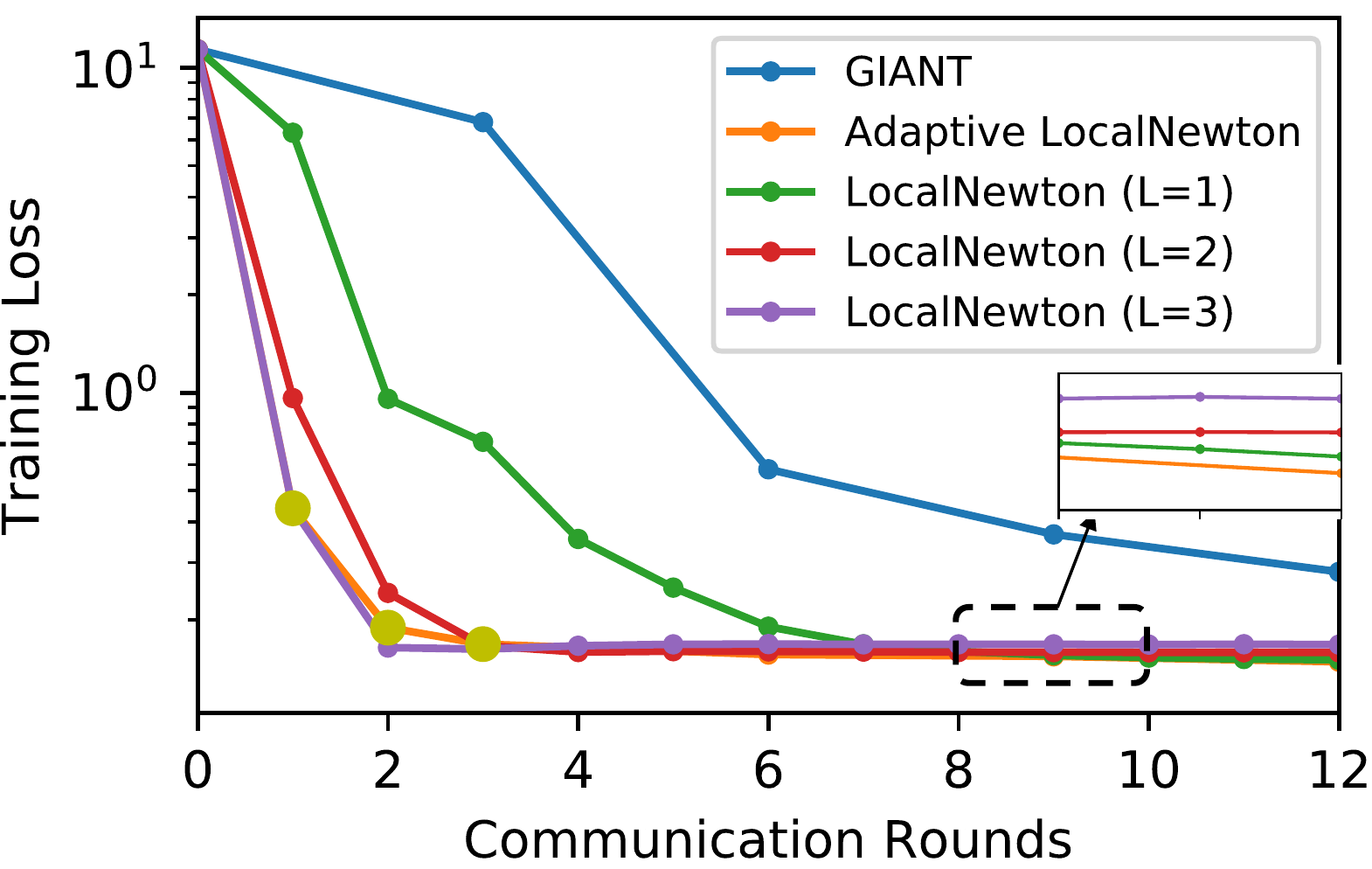}
        %\caption{50 workers}
        \caption{w8a dataset}
    \end{subfigure}
    ~
%     \begin{subfigure}{.31\textwidth}
%       \centering
%     \includegraphics[scale=0.45]{tex_figs/TrLoss_cov_t1.pdf}
%     %\caption{200 workers}
%     \end{subfigure}
% ~
    \begin{subfigure}{.45\textwidth}
      \centering
    \includegraphics[scale=0.66]{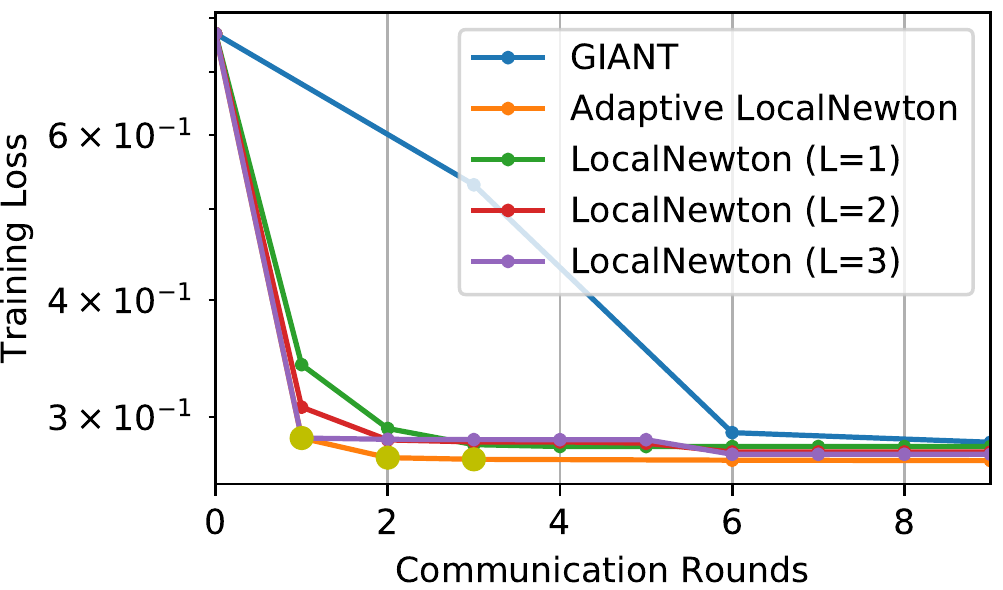}
    %\caption{200 workers}
    \caption{EPSILON dataset}
    \end{subfigure}

    % \begin{subfigure}{.45\textwidth}
    %     \centering
    %     \includegraphics[scale=0.45]{tex_figs/TeAcc_w8a_t1.pdf}
    %     \caption{w8a dataset}
    % \end{subfigure}
    % ~
%     \begin{subfigure}{.31\textwidth}
%       \centering
%     \includegraphics[scale=0.45]{tex_figs/TeAcc_cov_t1.pdf}
%     \caption{Covtype dataset}
%     \end{subfigure}
% ~~~~~
    % \begin{subfigure}{.45\textwidth}
    %   \centering
    % \includegraphics[scale=0.45]{tex_figs/TeAcc_eps_t1.pdf}
    % \caption{EPSILON dataset}
    % \end{subfigure}
    \caption{Comparing \ln~(for different values of $L$) and GIANT. In general, \ln~reaches very close to the optimal solution. In that region, however, the convergence rate of \ln~is slow. This is mitigated by using Adaptive \ln~which appends the \ln~iterations with better quality (but more expensive) updates from GIANT.}
    \label{fig:L123_giant}
\end{figure*}

{\bf Motivating Example (Least-squares).} Let us now consider the simple example of unregularized linear least squares, i.e., the loss function at the $k$-th worker is $f^k(\w) = \frac{1}{s}\|\y^k - \X^k \w\|^2$, where $\y^k = [y_1^k,y_2^k,\ldots,y_s^k]^\top \in \mathbb{R}^s$ and $\X = [\x_1^k,\ldots,\x_s^k]^\top \in \mathbb{R}^{s \times d}$. Note that one second-order iteration (with step-size one) reaches the optimal solution, say $(\w^k)^* = \w^k_t - (\H^k_t)^{-1}\g^k\w^k_t = [(\X^k)^T\X^k]^{-1}(\X^k)^T\y^k ~\forall~\w_t^k\in\R^{d}$, for the local loss function $f^k(\w)$ at the $k$-th worker.

Thus, applying \ln~here with $L\geq 1$ would imply that the local iterates at the $k$-th worker are fixed at $\w^k_t = (\w^k)^*$ while the global iterates at the master are fixed at $\bw_t = \frac{1}{K}\su(\w^k)^*$ for all $t\in [T]$. Note that $\bw_t \neq \bw^*$ in general, where $\bw^*$ is the optimal solution for the global problem in Eq. \eqref{mainProblem}. Hence, \ln~(Algorithm \ref{algo:localnewton}) does not reach the optimal solution for unregularized least-squares.
In fact, in Theorems~\ref{THM:L1} and \ref{THM:LGT1}, we show that running \ln~ algorithm results in an error floor of the order $1/\sqrt{s}$ for any convex loss function.
% Hence, \ln~ takes the iterates in a norm-ball of radius $\mathcal{O}(1/\sqrt{s})$ around the optima, but we have no further guarantee that \ln~hits the optimal solution. As seen in Theorems~\ref{THM:L1} and \ref{THM:LGT1}, this error-floor phenomenon happens for any convex losses. 

Motivated from the above example, if we want to attain the optimal solution, one needs to switch to optimization algorithms which yield no error floor. One standard example of such an algorithm is GIANT~(\cite{giant_nips}), which is a communication-efficient distributed second-order algorithm. However, for general convex functions, the convergence of GIANT requires the initial point to be close to the optimal solution. From this point of view, one can think of \ln~as an initialization scheme that pushes the iterates close to the solution (within a radius of $\mathcal{O}(1/\sqrt{s})$) with a few rounds of communication. Then, one can switch to the GIANT algorithm to obtain convergence to the solution point. We call this algorithm Adaptive \ln, where the master modulates the value of $L$ successively over iterations and finally switches to GIANT to obtain the final solution. The details are given in Algorithm~\ref{algo:adap_localnewton}.

Recall that GIANT synchronizes the local gradients and the local descent direction in every iteration \citep{giant_nips}. Further, it finds the step-size by doing a distributed backtracking line-search requiring an additional round of communication (Sec. 5.2, \cite{giant_nips}). Finally, the master updates the model by using the average descent direction and the obtained step-size and ships the model to all the workers. Thus, each iteration in GIANT requires three rounds of communication. 
This approach has compared favorably to other popular distributed second-order methods (e.g., DANE \cite{dane}, AGD \cite{nesterov_book}, BFGS \cite{bfgs_fletcher2013practical}, CoCoA \cite{cocoa}, DiSCO \cite{disco}). 

% \michal{This discussion of experiments should probably be in the experiments section.}
In Fig. \ref{fig:L123_giant}, we compare GIANT to \ln, where \ln~is run with 100 workers for $L=1,2$ and $3$, for two datasets---w8a and EPSILON obtained from LIBSVM \citep{libsvm}.
% (see Sec. \ref{sec:experiments} for further details on datasets).\footnote{Additional experiments on two more datasets---a9a and ijcnn1---are provided in Appendix E.1.} 
Note that \ln~converges much faster with respect to communication rounds for all the three datasets since it communicates intermittently, i.e., once every few local second-order iterations (e.g., after 3 local iterations for $L=3$). Not shown here is that testing accuracy follows the same trends. Further, the quality of the final solution improves as we reduce $L$. However, it reaches extremely close to the optimal training loss, but it converges very slowly (or flattens out) after that.

% \michal{We need some pseudocode for Adaptive \ln.}
These empirical observations further motivate Adaptive \ln: a second-order distributed algorithm that adapts the number of local iterations as the training progresses and ultimately finishes with GIANT. 
This can be done by monitoring the objective function at the master, e.g., reduce $L$ if the loss stops improving (or switch to GIANT if $L=1$).%
\footnote{To further reduce the communication rounds and dependency on $L$, each worker can update the model for multiple values of $L$ and send the concatenated model updated to the master. The master can decide the right value of $L$ by evaluating the loss/accuracy for these different models.} 
See Algorithm \ref{algo:adap_localnewton}, where we provide the pseudo-code for Adaptive \ln. 
Whenever the global function value at the master does not decrease more than a constant $\delta$, we decrease the value of $L$ to improve the quality of second order estimate. 
% For all our experiments, and regardless of datasets, Adaptive \ln~proceeds as follows: It starts with $L=3$, then decreases $L$ by one after each communication round, and finally (after reaching $L=1$), it switches to GIANT.
%(which demonstrates the robustness of \ln~to the exact value of $L$).
In this sense, Adaptive \ln~can be seen as providing a carefully-constructed or gradually-annealed initialization for GIANT.

\emph{Comparison with GIANT}: Algorithm~\ref{algo:adap_localnewton} is further motivated by the theoretical guarantees we obtain in the subsequent section. In Theorems \ref{THM:L1} and \ref{THM:LGT1} of Sec. \ref{sec:theoretical_guarantees}, we prove the convergence of \ln~to the optimal solution within an error floor starting with any initial point in $\R^d$. In sharp contrast, Theorem 2 in GIANT \citep{giant_nips}, the authors convergence guarantees to the  optimal solution when the current model is sufficiently close to the optimal model. 
From Fig.~\ref{fig:L123_giant}, we see that Adaptive \ln~significantly outperforms GIANT in terms of rounds of communication. Adaptive \ln~starts from L=3, and yellow dots in its plot denote the reduction in the value of $L$ by one or a switch to GIANT if $L=1$.

\section{Convergence Guarantees}
\label{sec:theoretical_guarantees}
In this section, we present the main theoretical contributions of the paper. For this, we only consider the (non-adaptive)  LocalNewton algorithm. Obtaining theoretical guarantees for Adaptive LocalNewton is kept as an interesting future work.

First, we delineate some assumptions on $f(\cdot)$ required to prove theoretical convergence of the proposed method.

{\bf Assumptions}: We make the following standard assumptions on the objective function $f(\cdot)$ for all $\w \in \R^d$:
\begin{enumerate}
    \item $f_i(\cdot)$, for all $i\in[n]$, is twice differentiable.
    \item $f(\cdot)$ is $\kappa$-strongly convex, that is,
        $\nabla^{2} f(\mathbf{w}) \succcurlyeq \kappa \mathbf{I}$.
    \item $f(\cdot)$ is $M$-smooth, that is, $\nabla^{2} f(\mathbf{w}) \preccurlyeq M \mathbf{I}$.
    \item $\|\nabla^2 f_i(\cdot)\|_2, i\in [n],$ is upper bounded. That is, $\nabla^2 f_i(\w) \preccurlyeq B\I$, for all $i\in[n]$. 
\end{enumerate} 

In the following lemma, we make use of matrix concentration inequalities to show that, for sufficiently large sample size $s$, the local Hessian at each worker is also strongly convex and smooth with high probability.

%%%%%%%%%% GIANT Lemma %%%%%%%%%%%%%%%%%
\begin{lemma}\label{lemma:subsampling_guarantee1}
Let $f(\cdot)$ satisfy assumptions 1-4 and $0 <\epsilon \leq 1/2$ and $0 < \delta <1$ be fixed constants. Then, if $s \geq \frac{4B}{\kappa\epsilon^2}\log\frac{2d}{\delta}$, the local Hessian at the $k$-th worker satisfies
\begin{align}
(1-\epsilon)\kappa \preccurlyeq \nabla^2 f^k(\w) = \H^k(\w) \preccurlyeq (1 + \epsilon)M,
\end{align}
for all $\w\in\R^d ~\text{and}~ k \in [K]$ with probability at least $1 - \delta$.
\end{lemma}
\begin{proof}
See Appendix \ref{app:aux_lemmas}.
\end{proof}

{\bf Step-size selection:} Let each worker locally choose a step-size according to the following rule
\begin{align}
    \alpha_{t}^k &=\max _{\alpha \leq \as} \alpha \quad \text { such that } \nn\\ 
    f^k\left(\mathbf{w}_{t}^k - \alpha \mathbf{p}_{t}^k\right) &\leq f^k\left(\mathbf{w}_{t}^k\right) - \alpha \beta \mathbf({\p}_{t}^k)^{T} \nabla f^k\left(\mathbf{w}_{t}^k\right),\label{local-ss}
\end{align}
for some constant $\beta \in (0,1/2]$, where the parameter $\as (\leq 1)$ depends on the properties of the objective function:%
\footnote{We introduce $\as$ here to establish theoretical guarantees. In our empirical results, we use the Armijo backtracking line-search rule with $\as=1$ (e.g., see \cite{boyd}) to find the right step-size.}
\begin{align}
\label{as_condition}
	\as \leq \min\left\{\frac{(1-\beta)\kappa}{M}, \frac{2\beta\kappa^2}{3M[M - \kappa/4]}\right\}.
\end{align}

Now, we are almost ready to prove the main theorems---Theorem~\ref{THM:L1} discusses the case when $L = 1$, and Theorem~\ref{THM:LGT1} discusses the case when $L > 1$.
Before that, let us state the following auxiliary lemma which is required to prove the main theorems.

\begin{lemma}
\label{lemma:local_linear_convergence1}
Let the function $f(\cdot)$ satisfy assumptions 1-3, and suppose that step-size $\alpha_t^k$ satisfies the line-search condition in \eqref{local-ss}. Also, let $0<\epsilon < 1/2$ and $0<\delta<1$ be fixed constants. Moreover, let the sample size $s \geq \frac{4B}{\kappa\epsilon^2}\log\frac{2d}{\delta}$. Then, the LocalNewton update, defined in  Eq. \eqref{w_update}, at the $k$-th worker satisfies
\begin{align*}
     f^k(\w_{t+1}^k) - f^k(\w_t^k) \leq  - \psi ||\g_t^k||^2 ~\forall~k\in[K],
\end{align*}
with probability at least $1 - \delta$, where $\psi = \frac{\as\beta}{M(1+\epsilon)}$. 
\end{lemma}
\begin{proof}
See Appendix \ref{app:aux_lemmas}.
\end{proof}

We next use the result in Lemma \ref{lemma:local_linear_convergence1} to prove linear convergence for the global function $f(\cdot)$. We first prove guarantees for the $L=1$ case, where the models are communicated every iteration but the gradient is computed locally instead of globally contrary to previous results~\citep{giant_nips} (thus reducing two communication rounds per iteration). We then extend it to the general case of $L>1$ and show that the updates converge at a sublinear rate in that case.

\begin{theorem}
\label{THM:L1}
[$L=1$ case]
Suppose Assumptions 1-5 hold and the step-size $\alpha_t^k$ satisfies the line-search condition~\eqref{local-ss}. Also, let $0<\delta<1$,  $0<\epsilon, \epsilon_1<1/2$ be fixed constants and let $ \Gamma = \max_{1\leq i \leq n} \|\nabla f_i(.)\|$. Moreover, assume that the  sample size for each worker satisfies  $s \geq \frac{4B}{\kappa\epsilon^2}\log\frac{2dK}{\delta}$, where the samples are chosen without replacement. Then, with the LocalNewton updates, $\{ \bw_{t} \}_{t \geq 0}$, from Algorithm~\ref{algo:localnewton} and $L=1$, we obtain
\begin{enumerate}
    \item  If $s \gtrsim  \frac{\Gamma^2}{\epsilon_1^2 G^2} \log (d/\delta)$ for $G=\min_k\|g^k(\bw_t)\|$,
    we get with probability at least $1-6K\delta$,
    %Provided $\|g^k(\bw_t)\| \geq G$  for all $k\in [K]$, and $ s \gtrsim  \frac{\Gamma^2}{\epsilon_1^2 G^2} \log (d/\delta)$, we get
    \begin{align*}
f(\bw_{t+1}) - f(\w^*) \leq \rho_1(f(\bw_t) - f(\w^*)).
\end{align*}
%with probability exceeding $1-6K\delta$.
\item We obtain, with probability at least $1-6K\delta$,
\begin{align*}
% f(\bw_{t+1}) - f(\w^*) 
% &\leq \rho_2(f(\bw_t) - f(\w^*))
% \\
% &\quad +\frac{1}{\sqrt s}\cdot\frac{\Gamma^2(1+ \sqrt{2\log (\tfrac{m}{\delta})})}{\kappa(1-\epsilon)}.
f(\bw_{t+1}) - f(\w^*)
\leq \rho_2(f(\bw_t) - f(\w^*))  + \eta\cdot \frac{\Gamma}{\kappa(1-\epsilon)},
\end{align*}
where $\eta = \frac1{\sqrt s}\,\Gamma(1+ \sqrt{2\log (\frac{1}{\delta})}) $.
\end{enumerate}
Here $\rho_i = (1 - 2\kappa C_i)$, for $i=\{1,2\}$, $C_1 = \frac{(1-\epsilon)\psi}{2} - \frac{\epsilon_1}{\kappa(1-\epsilon)} $, $C_2 = \frac{\psi(1-\epsilon)}{2}$, and $\psi = \frac{\as\beta}{M(1+\epsilon)}$. 
\end{theorem}
\begin{proof}
The proof is presented in Appendix \ref{app:thm1}. Here, we provide a sketch of the proof.
\begin{enumerate}
\itemsep0em 
    \item  
    Due to the uniform sampling guarantee from Lemma \ref{lemma:subsampling_guarantee1}, the strong-convexity and smoothness of the global function $f(\cdot)$ implies that the local function at the $k$-th worker, $f^k(\cdot)$, also satisfy similar properties. Using this, we can lower bound $f(\bw_t) - f(\bw_{t+1})$ in terms of $\frac{1}{K}\su f^k(\w_t^k) - f^k(\w_{t+1}^k)$. 
    \item 
    Apply Lemma \ref{lemma:local_linear_convergence1} (that is, the result for standard Newton step) which says $f^k(\w_t^k) - f^k(\w_{t+1}^k) \geq \psi ||\g_t^k||^2~\forall~k\in[K]$.
    \item 
    Using uniform sketching argument, local gradients $g^k(\bw_t)$ are close to global gradient $g(\bw_t)$.
\end{enumerate}
\end{proof}

Some remarks regarding the convergence guarantee in Theorem \ref{THM:L1} are in order.
\begin{remark}
The above theorem implies that for $L=1$, the convergence rate of \ln~  is linear with high probability. Choosing $\delta = 1/\mathsf{poly}(K)$, we obtain the high probability as $1-1/\mathsf{poly}(K)$.
\end{remark}
\begin{remark}
We have two different settings in the above theorem. 
The first setting implies that provided the local gradients $\{g^k(\bw_t)\}_{k=1}^K$ are large enough, and the amount of local data $s$ is reasonably large, then the convergence is \emph{purely} linear and does not suffer an error floor. 
This will typically happen in the earlier iterations of \ln. 
% This condition could be restrictive. If this is violated, we move to the next setting.
Note that the gradients vanish as we get closer to the optimum, which is why the setting will eventually be violated. If this happens, we move to the next setting.
\end{remark}
\begin{remark}
The second setting implies that, if the gradient condition and the restriction of $s$ are violated, although the convergence rate of \ln~is still linear, the algorithm incurs an error floor. However, in this setting, the error floor is $\mathcal{O}(1/\sqrt{s})$, and hence it is quite small for sufficiently large sample-size, $s$, at each worker.
\end{remark}

To understand the linear convergence rate of \ln, we consider the following example.
Assume all the workers initialize at $\bw_0$ and run \ln~with $L=1$ for $T$ iterations and the first setting is true (that is, $\|\g^k(\w_t)\| \geq G$ for all $k\in [K]$ and $t \in [T]$). 
 Then, from Theorem \ref{THM:L1}, to reach within $\xi$ of the optimal function value (that is, $f(\bw_{T}) - f(\bw^*) \leq \xi$), the number of iterations $T$ is upper bounded by 
$$T \leq \left(\log\frac{1}{\rho_1}\right)\log\frac{\xi}{f(\bw_0) - f(\bw^*)}$$ 
with probability $1-6K\delta$ for a sample size $s \geq \max\left\{ \frac{4B}{\kappa\epsilon^2}\log\frac{2dKT}{\delta}, \frac{\Gamma^2}{\epsilon_1^2 G^2} \log \frac{dT}{\delta}\right\}$. 
(Note the increase in sample size $s$ by a factor of $T$ in the $\log(\cdot)$ due to a union bound). The fully synchronized second order method GIANT \citep{giant_nips} also has similar linear quadratic convergence but it assumes that the gradients are synchronized in every iteration.
We  remove this assumption by tracking how far the iterate deviates when the gradients are computed locally, thereby cutting the communication costs in half while still showing linear convergence (within some error floor in the most general~case). 

We now prove convergence guarantees for the case when $L>1$ in this algorithm. %Thus, theorem \ref{THM:L1} solves the open problem from \cite{giant_nips} by proving convergence guarantees for the case when gradients are not synchronized. Next, we prove convergence guarantees for the case when $L>1$.

% \michal{There needs to be a discussion saying why this bound suggests some convergence?}
\begin{theorem}[$L \geq 1$ case]\label{THM:LGT1}
Suppose Assumptions 1-4 hold and step-size $\alpha_t^k$ solves the line-search condition~\eqref{local-ss}. Also, let $0<\delta<1$, $0<\epsilon<1/2$ be fixed constants and let $ \Gamma = \max_{1\leq i \leq n} \|\nabla f_i(.)\|$. Moreover, assume that the  sample size for each worker satisfies  $s \geq \frac{4B}{\kappa\epsilon^2}\log\frac{2dK}{\delta}$, where the samples are chosen without replacement. Then, the LocalNewton updates, $\{ \bw_{t} \}_{t \geq 0}$, from Algorithm~\ref{algo:localnewton} and $L \geq 1$, with probability at least $1-6LK\delta$, satisfy
\begin{align*}
 f(\bw_{t+1}) - f(\bw_{t_0})  
 &\leq  -C\sum_{\tau = t_0}^t \left(\frac{1}{K}\su \|\g_\tau^k\|^2\right) 
 +  \eta\cdot \frac{L \Gamma}{\kappa(1-\epsilon)},
  \end{align*}
%  }
where  $\eta = \frac1{\sqrt s}\,\Gamma(1+ \sqrt{2\log (\frac{1}{\delta})})$, $C = \psi - \frac{(M - \kappa(1-\epsilon)^2)}{2K\kappa^2(1-\epsilon)^2} $. 
where $t_0$ is the last iteration where the models were synced, $\psi = \frac{\as\beta}{M(1+\epsilon)}$, and $C = \frac{\psi(1-\epsilon)^3}{2}.$
\end{theorem}
\begin{proof}
% The general idea of the proof follows the proof for the easier case in Theorem \ref{THM:L1} (i.e., when $L=1$).
The proof is presented in Appendix \ref{app:thm2}.
\end{proof}
\begin{remark}
The theorem shows that \ln\ with high probability produces a descent direction, provided that the error floor is sufficiently small, i.e. for sufficiently large $s$ (since $\eta$ is proportional to $1/\sqrt{s}$). %Since the loss function is non-negative (bounded from below), this implies that  \ln~ produces a descent direction. 
% \michal{The word "converges" is misleading here. Even without the error floor, we don't know that the algorithm converges to the optimum (it could converge to a suboptimal point, right?). I would stick to saying that it "produces a descent direction".}. 
Observe that the convergence rate here is no longer linear. In other words, we are trading-off the rate of convergence for local iterations ($L>1$).
\end{remark}
\begin{remark}
Choosing $\delta = 1/\mathsf{poly}(K,L)$, we get that the theorem holds with probability at least $1-1/\mathsf{poly}(K,L)$. Note that this is not restrictive since the dependence on $\delta$ is logarithmic. 
\end{remark}

% The proof is presented in Section \ref{app:thm2}. 
% A few comments on the Theorems \ref{THM:L1} and \ref{THM:LGT1} are in order.
%In Theorem~\ref{THM:L1} we prove a linear convergence of the function value to the global minimum, with
%high probability. 
%Combining this with the strong convexity of the function $\f$, we also have that the iterate sequence 
%$\{ \bw^{t} \}_{t \geq 0}$ convergences to the unique optimum $\bw^*$ at a linear rate. 
%Theorem~\ref{THM:LGT1} proves that the norm gradient, when averaged over the workers, converges to the 
%zero at a sublinear rate \michal{What does that mean?}; this combining with the convexity of the function $\f$ 
%ensures the global convergence to the optimal solution. \michal{From the statement, it is not obvious how this theorem implies any sort of convergence. Can we write the statement so that it's clear?} 
While the theoretical guarantees for $L>1$ 
in Theorem~\ref{THM:LGT1} are not as strong as those for $L=1$ in Theorem~\ref{THM:L1} (linear versus sublinear convergence), empirically we observe a fast rate
of convergence even when $L>1$ (see Section \ref{sec:experiments} and Appendix \ref{app:additional_exps}
for empirical results). Nevertheless, to the best of our knowledge, 
Theorem~\ref{THM:LGT1} is the first to show a descent guarantee for a distributed second-order method without synchronizing at every iteration.
%Theorem~\ref{THM:LGT1} is the first guarantee of its kind that shows global convergence to the optimal without synchronizing in every iteration. 
Obtaining a better rate of convergence for general $L$, with or without error floor,
is an interesting and relevant future research direction.

\section{Empirical Evaluation}\label{sec:experiments}

In this section, we present an empirical evaluation of our approach when solving a large-scale logistic regression problem. 
We ran our experiments on AWS Lambda workers using the PyWren \citep{pywren} framework. AWS Lambda is a \emph{serverless} computing platform which uses a high-latency cloud storage (AWS S3) to exchange data with the workers. The master is a \emph{serverful} AWS EC2 machine of type {\tt m4.4xlarge} which co-ordinates with the serverless workers and holds the central model during the entire training run. 
We ran experiments on the real-world datasets described in Table \ref{table:datasets} (obtained from LIBSVM \citep{libsvm}).

 \begin{table}[ht] 
 \centering 
 \begin{tabular}{cccc} 
 \toprule
\multicolumn{1}{c}{Dataset}  &
% \multicolumn{1}{c}{\bf Training Samples ($n$)} &
\makecell{Training \\ samples ($n$)} &
\makecell{Features \\ ($d$)} &
% \multicolumn{1}{c}{\bf Testing samples}  \\
\makecell{Testing \\ samples} \\
% \hline CIFAR10 & Test Accuracy (\%) & Training Time (sec) \\
 \hline  
 w8a & $48,000$ & $300$ & $15,000$ \\ 
 Covtype & $500,000$ & $2916$ & $81,000$  \\ 
 EPSILON & $400,000$ & $2000$ & $100,000$ \\ 
 a9a & $32,000$ & $123$ & $16,000$ \\ 
 ijcnn1 & $49,000$ & $22$ & $91,000$   \\ 
\bottomrule
 \end{tabular} 
  \caption{ Datasets considered for experiments in this paper} 
  \label{table:datasets} 
 \end{table}

 \begin{figure*}[ht]
\centering
\begin{subfigure}{.31\textwidth}
      \centering
    \includegraphics[scale=0.45]{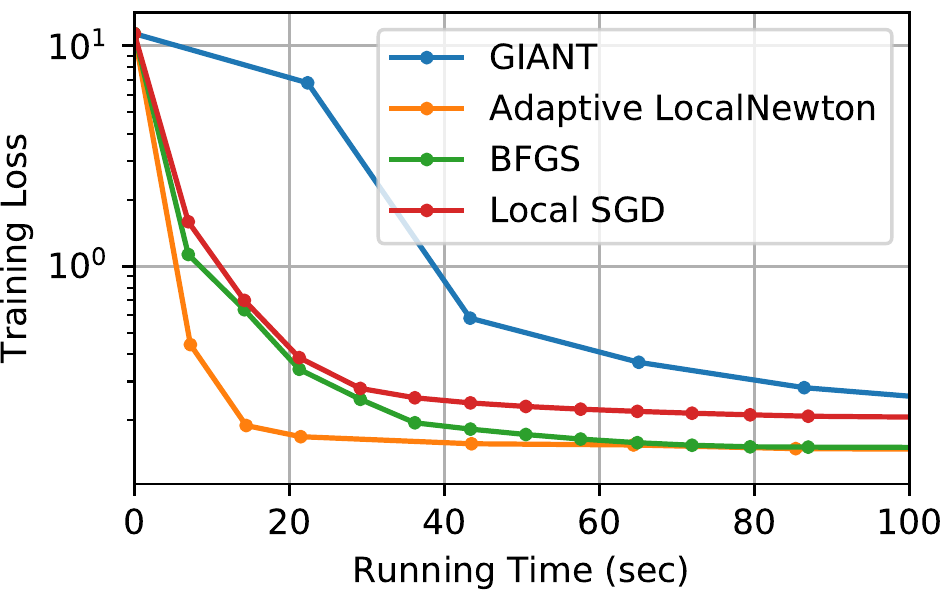}
    %\caption{w8a dataset: Training loss with 100 workers}
    \end{subfigure}
    ~
    \begin{subfigure}{.31\textwidth}
        \centering
        \includegraphics[scale=0.45]{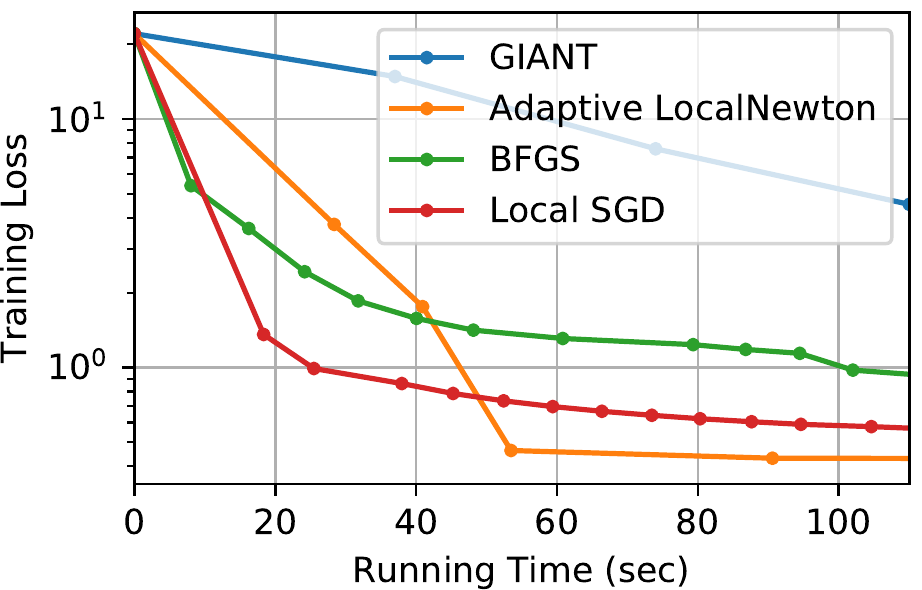}
    \end{subfigure}
    ~
    \begin{subfigure}{.31\textwidth}
        \centering
        \includegraphics[scale=0.45]{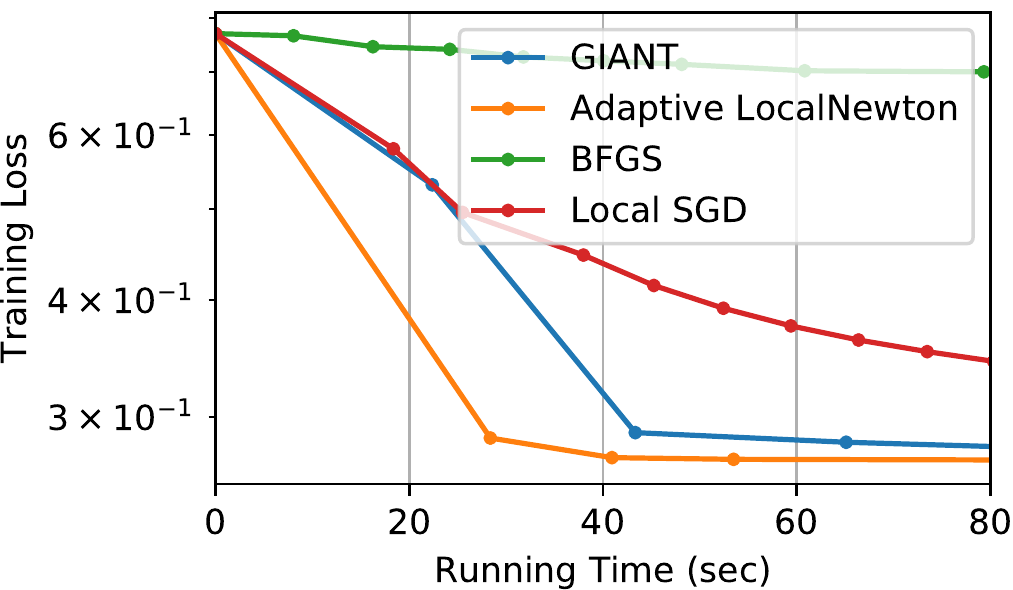}
    \end{subfigure}

    \begin{subfigure}{.31\textwidth}
        \centering
        \includegraphics[scale=0.45]{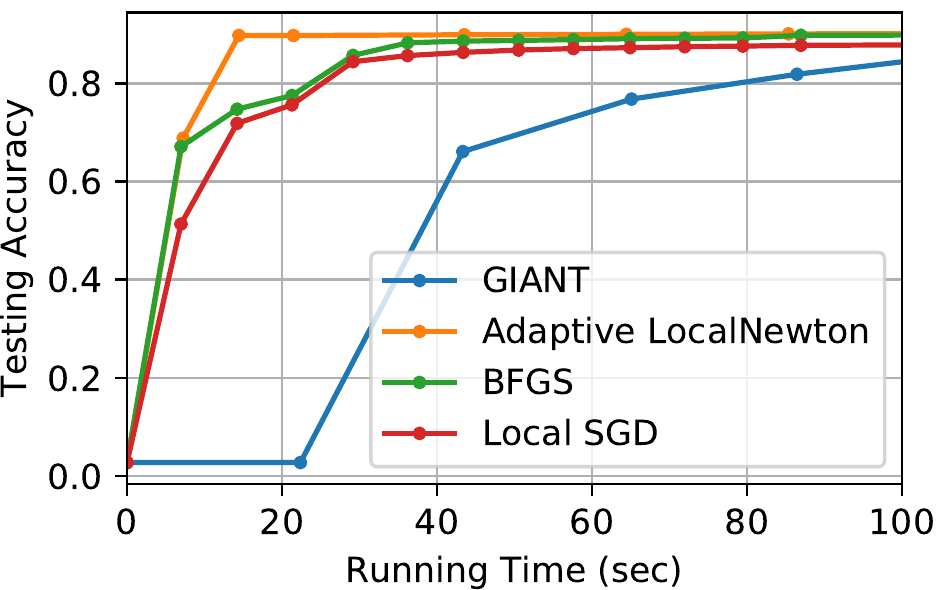}
        \caption{w8a dataset}
    \end{subfigure}
    ~
    \begin{subfigure}{.31\textwidth}
      \centering
    \includegraphics[scale=0.45]{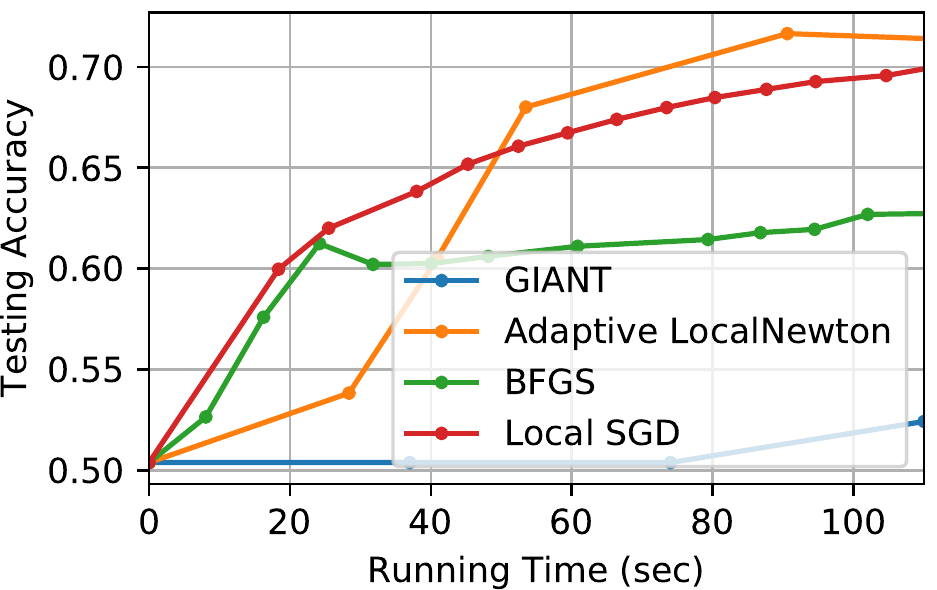}
    \caption{Covtype dataset}
    \end{subfigure}
    \hspace{4mm}
    \begin{subfigure}{.31\textwidth}
      \centering
    \includegraphics[scale=0.45]{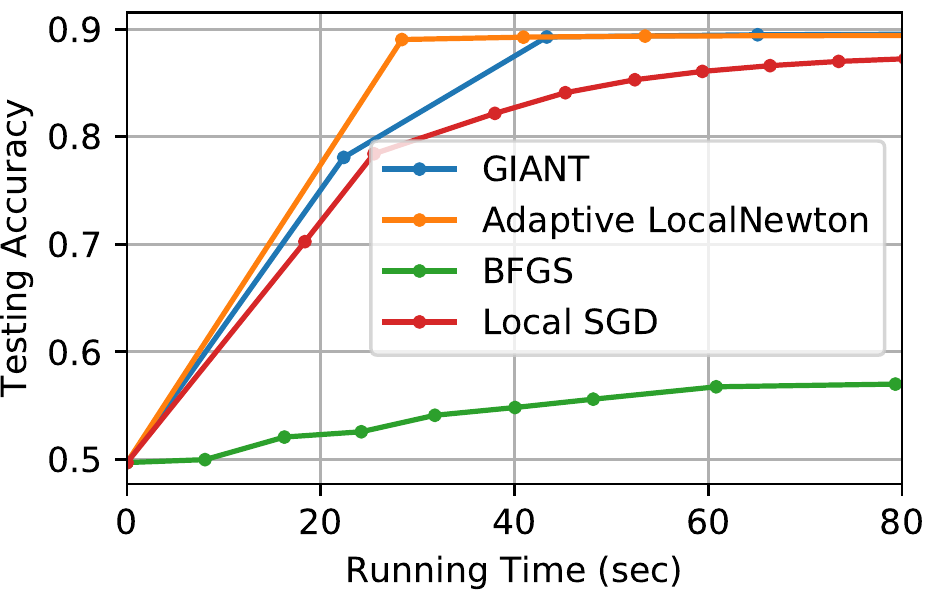}
    \caption{EPSILON dataset}
    \end{subfigure}    
    
    \begin{subfigure}{.46\textwidth}
        \centering
        \includegraphics[scale=0.45]{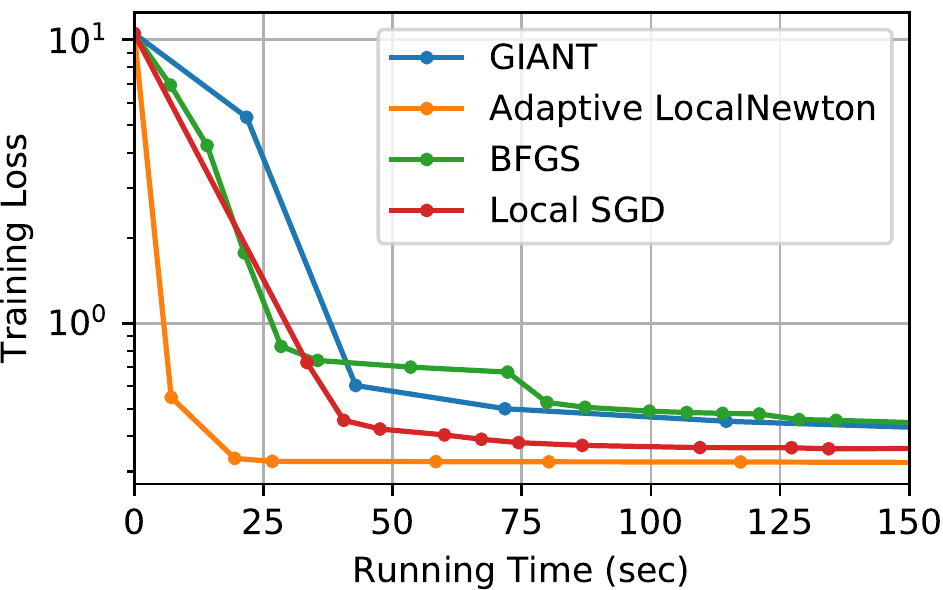}
        %\caption{50 workers}
    \end{subfigure}
    ~
\begin{subfigure}{.46\textwidth}
        \centering
        \includegraphics[scale=0.45]{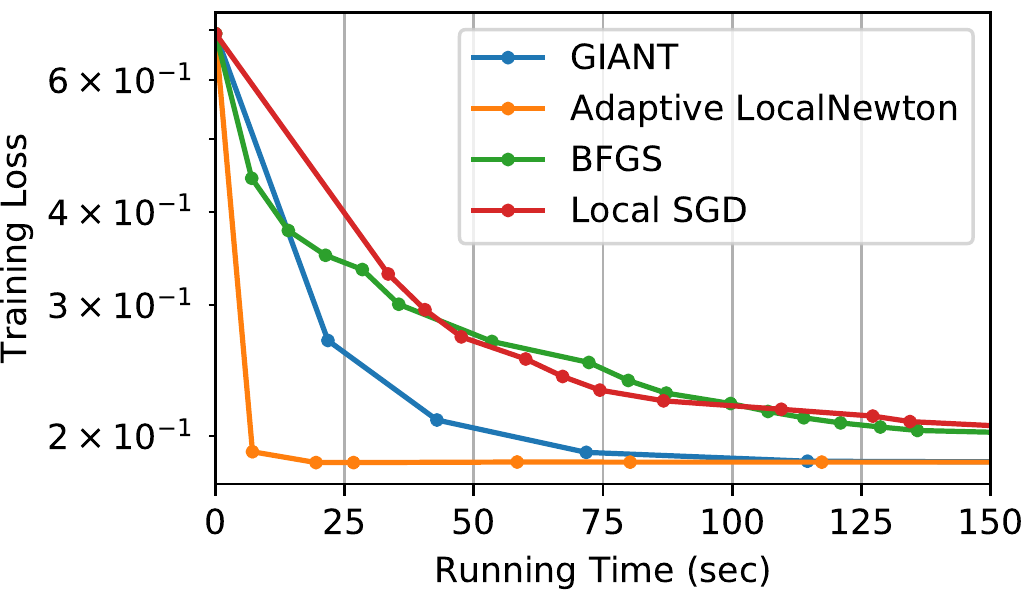}
        %\caption{50 workers}
    \end{subfigure}
    
    \begin{subfigure}{.46\textwidth}
      \centering
    \includegraphics[scale=0.45]{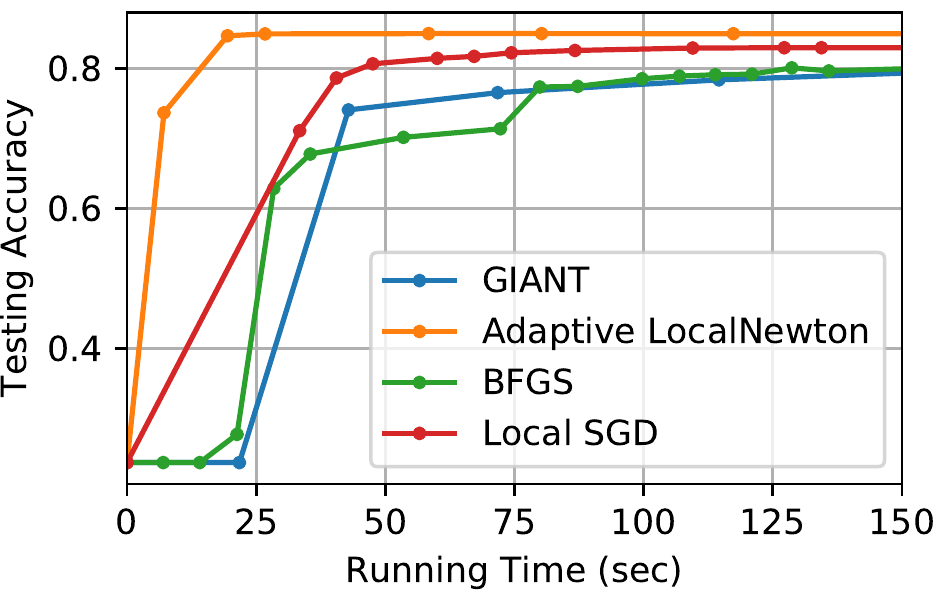}
    \caption{a9a dataset}
    \end{subfigure}
\hspace{3mm}
    \begin{subfigure}{.46\textwidth}
      \centering
    \includegraphics[scale=0.45]{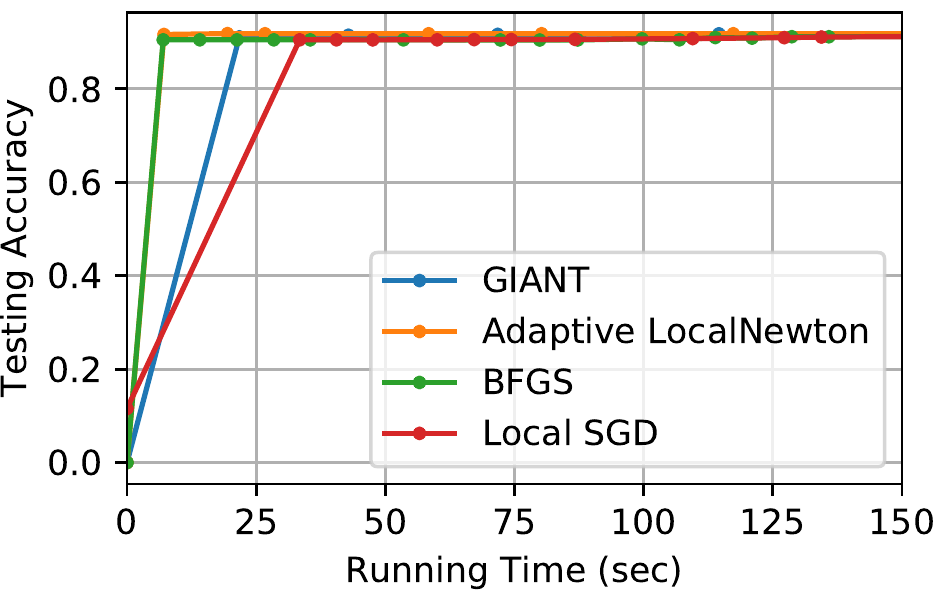}
    \caption{ijcnn1 dataset}
    \end{subfigure}
    
    \caption{Experiments on the the different datasets from Table \ref{table:datasets} on AWS Lambda. Both in terms of training loss and testing accuracy, Adaptive \ln~converges to the optimal value at least $50\%$ faster than existing schemes.}
    \label{fig:serverless_times_main}
\end{figure*}

\begin{figure*}[ht]
\centering
    \begin{subfigure}{.31\textwidth}
        \centering
        \includegraphics[scale=0.45]{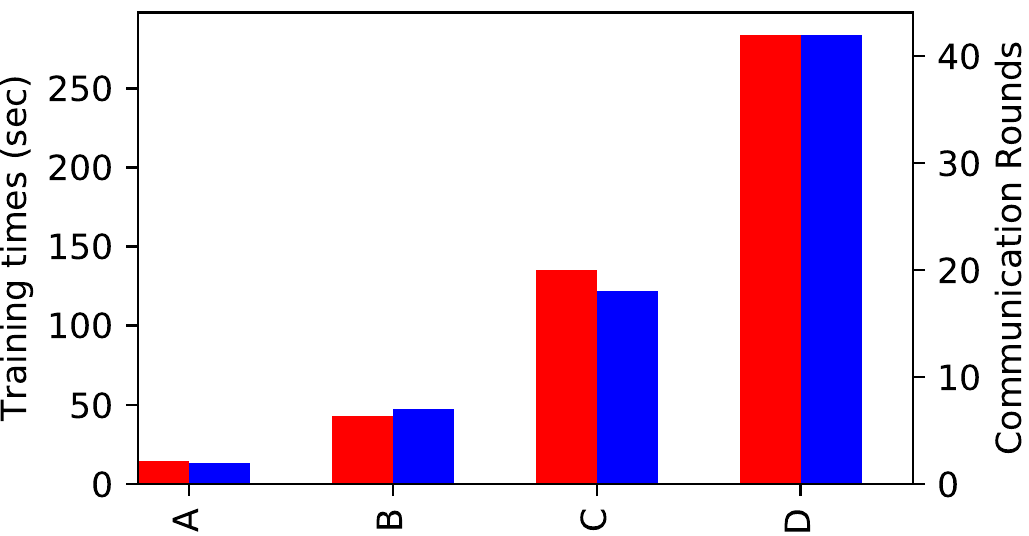}
        \caption{w8a dataset}
    \end{subfigure}
    ~
    \begin{subfigure}{.31\textwidth}
      \centering
    \includegraphics[scale=0.45]{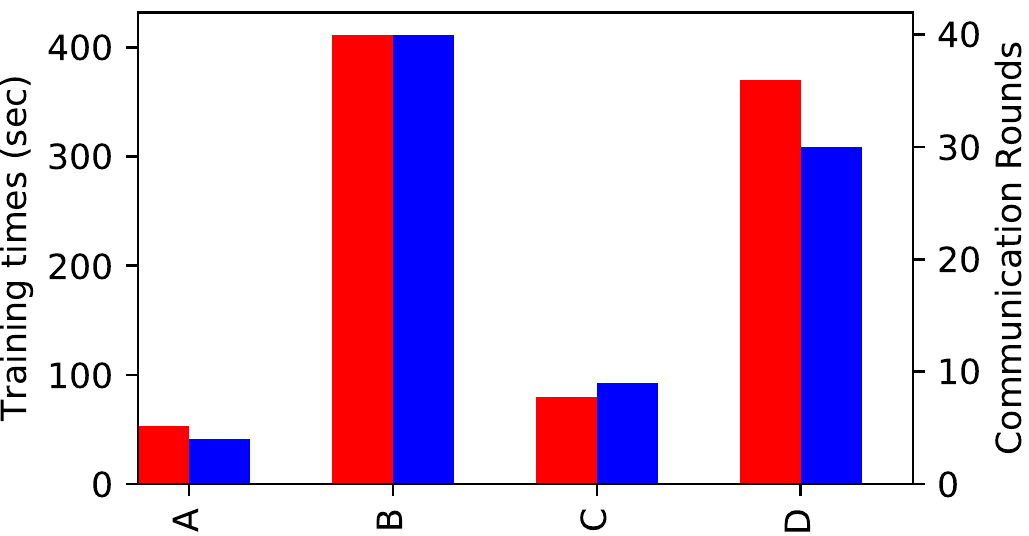}
    \caption{Covtype dataset}
    \end{subfigure}
    ~
    \begin{subfigure}{.31\textwidth}
      \centering
    \includegraphics[scale=0.45]{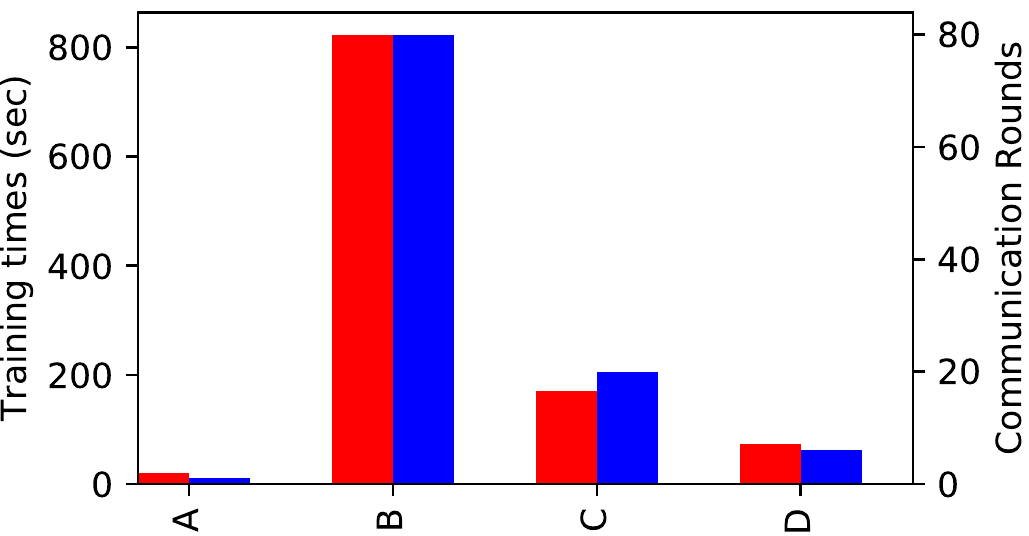}
    \caption{EPSILON dataset}
    \end{subfigure}    
    \caption{Training times (red bars) and communication rounds (blue bars) required to reach the same training loss of $0.19, 0.65$ and $0.3$ for the w8a, Covtype and EPSILON datasets, respectively, on AWS Lambda. Here, A: Adaptive \ln, B: BFGS, C: Local SGD, D: GIANT.}
    \label{fig:bar_plots}
\end{figure*}

We compare the following distributed optimization schemes for the above datasets:
\begin{enumerate}
    \item Local SGD \citep{stich2018local}: The workers communicate their models once every epoch, where training on one epoch implies applying SGD (with mini-batch size one) over one pass of the dataset stored locally at the worker. The best step-size was obtained through hyperparameter tuning (Table \ref{table:hyps}). 
\item BFGS \citep{bfgs_fletcher2013practical}: BFGS is a popular quasi-Newton method that estimates an approximate Hessian from the gradient information from previous iterations. The step-size was obtained using backtracking line-search.
The best step-size was obtained through hyperparameter tuning (Table \ref{table:hyps}).
% (see Appendix E.2 for details).\\
\item GIANT \citep{giant_nips}: A state-of-the-art distributed second order algorithm proposed in \cite{giant_nips}. The authors show that GIANT outperforms many popular schemes such as DANE, AGD, etc.
The step-size was obtained using distributed line-search as described in \cite{giant_nips}.
\item Adaptive \ln: For all the considered datasets, Adaptive \ln~gradually reduces $L$ if the loss function stops decreasing, starting from $L=3$ in the first round of communication. In general, during the later stages of optimization, it switches to GIANT owing to its  convergence to the optimal solution when $\bw_t$ is sufficiently close to $\w^*$.
The step-size was obtained using backtracking line-search locally at each worker as described in Algorithm \ref{algo:localnewton}.
\end{enumerate}

We also implemented the distributed second-order optimization scheme from \cite{pmlr-v80-duenner18a} with moderate hyperparameter tuning, and observed that its performance is either comparable or worse than the baseline GIANT \cite{giant_nips}. Hence, for clarity, we omit those results from our plots.

For all the experiments presented in this paper, we fixed the number of workers, $K$, to be $100$. Hence, the number of samples per worker, $s = n/100$, for all datasets. The regularization parameter was chosen to be $\gamma = 1/n$.
Note that there are several other schemes--such as AGD \citep{nesterov_book}, DANE \citep{dane} and SVRG \cite{svrg}--that have been proposed in the literature for communication-efficient optimization. However, most of these schemes have been shown to be outperformed by one of Local SGD, BFGS or GIANT, and hence, we do not perform the comparison again. 
% \vip{Verify this} 
In Fig. \ref{fig:serverless_times_main}, we plot the training loss and testing accuracy for w8a, covtype\footnote{The covtype dataset has $d=54$ features and it does not perform well with logistic regression. Hence, we apply polynomial feature extension (using pairwise products) to increase the number of features to $d^2 = 2916$.}, EPSILON, a9a and ijcnn1 datasets. For all the datasets considered, Adaptive \ln~significantly outperforms its competitors in terms of time required to reach the same training loss (or testing accuracy) on AWS Lambda. Furthermore, since the number of workers, $K$, is fixed and serverless platform charges are proportional to the total CPU hours, end-to-end training time is directly proportional to the costs charged by AWS for training. This is assuming that each worker roughly takes the same amount of time per iteration.

In Fig. \ref{fig:bar_plots}, we highlight the fact that runtime savings on AWS Lambda are a direct consequence of significantly fewer rounds of communication. Specifically, to reach the same training loss, we plot the training times and communication rounds as bar plots for three datasets, and we note that savings in communication rounds results in commensurate savings on end-to-end runtimes on AWS Lambda. 
In Fig. \ref{fig:comm_rounds_all} in Appendix \ref{app:additional_exps}, we provide detailed plots for training losses and testing accuracies with respect to communication rounds for all the five datasets.

 \begin{table}[h] 
 \centering 
 \begin{tabular}{cccc} 
\multicolumn{1}{c}{\bf Dataset}  &
\multicolumn{1}{c}{\bf Samples per worker (s)} &
\multicolumn{1}{c}{\bf Local SGD} &
\multicolumn{1}{c}{\bf BFGS}  \\
% \hline CIFAR10 & Test Accuracy (\%) & Training Time (sec) \\
 \hline  
 w8a & $480$ & $10/s$ & $100$ \\ 
 Covtype & $5000$ & $10/s$ & $1$  \\ 
 EPSILON & $4000$ & $500/s$ & $10$ \\ 
 a9a & $320$ & $10/s$ & $1$ \\ 
 ijcnn1 & $490$ & $100/s$ & $10$   \\ 
 \hline 
 \end{tabular} 
 \caption{Step-sizes obtained using tuning for Local SGD and BFGS for several datasets} 
 \label{table:hyps} 
 \end{table}
 
 {\bf Hyperparameter details}:
In Table \ref{table:hyps}, we provide the step-sizes for local SGD and BFGS that were obtained through hyperparameter tuning, where $s = n/K$, $n$ is the number of training examples in the dataset and $K=100$.

\section{Conclusion}

The practicality of second-order optimization methods has been questioned since naive ways to implement them require large compute and power storage to work with the Hessian. 
However, in the last few decades, trends such as Moore's law have made computation faster and memory cheaper, while improvements in communication costs have been at best marginal. 
These trends, combined with a flurry of efficient but approximate algorithms \citep{mert,fred1,fred2,osn}, have revived interest in second-order methods.
In this paper, we identify and concretize the role that second-order methods---combined with local optimization algorithms---can play in reducing the communication costs during distributed training, in particular in serverless environments. 
Since second-order information has recently been used to develop state-of-the-art methods for deep neural networks with extremely large model sizes \citep{hawq,yao2019pyhessian,qbert, YGSKM20_adahessian_TR}, we expect that methods such as ours will play a significant role in motivating and designing next-generation communication-efficient algorithms for fast distributed training of machine learning models.

\section*{Acknowledgments}
  This work was partially supported by NSF grants CCF-1704967 and CCF- 0939370 (Center for Science of Information) to VG; NSF Grant CCF-2007669 and CCF-1703678 to AG and KR; and ARO, DARPA, NSF, and ONR grants to MD, RK and MWM. The authors would like to additionally thank Koulik Khamaru for helpful discussions with the proof and to AWS for providing promotional cloud credits for research. 

\bibliographystyle{ieeetr}
\bibliography{bibli}

\begin{thebibliography}{10}

\bibitem{demmel2}
J.~Demmel, ``Communication-avoiding algorithms for linear algebra and beyond,''
  in {\em 2013 IEEE 27th Int. Sym. on Parallel and Distributed Processing},
  pp.~585--585, May 2013.

\bibitem{oversketch}
V.~Gupta, S.~Wang, T.~Courtade, and K.~Ramchandran, ``Oversketch: Approximate
  matrix multiplication for the cloud,'' {\em IEEE International Conference on
  Big Data, Seattle, WA, USA}, 2018.

\bibitem{sparse_comm_distributed2019}
J.~Acharya, C.~De~Sa, D.~Foster, and K.~Sridharan, ``Distributed learning with
  sublinear communication,'' in {\em International Conference on Machine
  Learning}, pp.~40--50, 2019.

\bibitem{stich2018sparsified}
S.~U. Stich, J.-B. Cordonnier, and M.~Jaggi, ``{Sparsified SGD with memory},''
  in {\em Advances in Neural Information Processing Systems}, pp.~4447--4458,
  2018.

\bibitem{sgd_with_sketching}
N.~Ivkin, D.~Rothchild, E.~Ullah, I.~Stoica, and R.~Arora,
  ``{Communication-efficient distributed SGD with sketching},'' in {\em
  Advances in Neural Information Processing Systems}, pp.~13144--13154, 2019.

\bibitem{konevcny2016federated}
J.~Kone{\v{c}}n{\`y}, H.~B. McMahan, F.~X. Yu, P.~Richt{\'a}rik, A.~T. Suresh,
  and D.~Bacon, ``Federated learning: Strategies for improving communication
  efficiency,'' {\em arXiv preprint arXiv:1610.05492}, 2016.

\bibitem{ghosh_one}
{Ghosh, Avishek}, {Maity, Rajkumar}, {Kadhe, Swanand}, {Mazumdar, Arya}, and
  {Ramachandran, Kannan}, ``Communication efficient and byzantine tolerant
  distributed learning,'' in {\em 2020 IEEE International Symposium on
  Information Theory (ISIT)}, pp.~2545--2550, 2020.

\bibitem{lin2017deep}
Y.~Lin, S.~Han, H.~Mao, Y.~Wang, and W.~J. Dally, ``Deep gradient compression:
  Reducing the communication bandwidth for distributed training,'' {\em arXiv
  preprint arXiv:1712.01887}, 2017.

\bibitem{bernstein2018signsgd}
J.~Bernstein, Y.-X. Wang, K.~Azizzadenesheli, and A.~Anandkumar, ``{signSGD:
  Compressed Optimisation for Non-Convex Problems},'' in {\em International
  Conference on Machine Learning}, pp.~560--569, 2018.

\bibitem{hawq}
Z.~Dong, Z.~Yao, A.~Gholami, M.~W. Mahoney, and K.~Keutzer, ``{HAWQ: Hessian
  aware quantization of neural networks with mixed-precision},'' in {\em
  Proceedings of the IEEE International Conference on Computer Vision},
  pp.~293--302, 2019.

\bibitem{qbert}
S.~Shen, Z.~Dong, J.~Ye, L.~Ma, Z.~Yao, A.~Gholami, M.~W. Mahoney, and
  K.~Keutzer, ``{Q-BERT: Hessian based ultra low precision quantization of
  BERT},'' in {\em Proceedings of the AAAI Conference on Artificial
  Intelligence}, vol.~34, pp.~8815--8821, 2020.

\bibitem{tyagi_limits_quantized2020}
P.~{Mayekar} and H.~{Tyagi}, ``{Limits on Gradient Compression for Stochastic
  Optimization},'' {\em arXiv e-prints}, p.~arXiv:2001.09032, jan 2020.

\bibitem{federated_learning_survey2019}
P.~Kairouz, H.~B. McMahan, B.~Avent, A.~Bellet, M.~Bennis, A.~N. Bhagoji,
  K.~Bonawitz, Z.~Charles, G.~Cormode, R.~Cummings, {\em et~al.}, ``Advances
  and open problems in federated learning,'' {\em arXiv preprint
  arXiv:1912.04977}, 2019.

\bibitem{pmlr-v54-mcmahan17a}
B.~McMahan, E.~Moore, D.~Ramage, S.~Hampson, and B.~A. y~Arcas,
  ``{Communication-Efficient Learning of Deep Networks from Decentralized
  Data},'' in {\em Proceedings of the 20th International Conference on
  Artificial Intelligence and Statistics}, vol.~54, pp.~1273--1282, 2017.

\bibitem{osn}
V.~Gupta, S.~Kadhe, T.~Courtade, M.~W. Mahoney, and K.~Ramchandran,
  ``{Oversketched Newton: Fast convex optimization for serverless systems},''
  {\em IEEE International Conference on Big Data, Seattle, WA, USA}, 2020.

\bibitem{hellerstein}
J.~M. Hellerstein, J.~Faleiro, J.~E. Gonzalez, J.~Schleier-Smith, V.~Sreekanti,
  A.~Tumanov, and C.~Wu, ``Serverless computing: One step forward, two steps
  back,'' {\em arXiv preprint arXiv:1812.03651}, 2018.

\bibitem{serverless_computing}
I.~Baldini, P.~C. Castro, K.~S.-P. Chang, P.~Cheng, S.~J. Fink, V.~Ishakian,
  N.~Mitchell, V.~Muthusamy, R.~M. Rabbah, A.~Slominski, and P.~Suter,
  ``Serverless computing: Current trends and open problems,'' {\em CoRR},
  vol.~abs/1706.03178, 2017.

\bibitem{berkeley_view}
E.~Jonas, J.~Schleier-Smith, V.~Sreekanti, C.-C. Tsai, A.~Khandelwal, Q.~Pu,
  V.~Shankar, J.~Carreira, K.~Krauth, N.~Yadwadkar, {\em et~al.}, ``Cloud
  programming simplified: A berkeley view on serverless computing,'' {\em arXiv
  preprint arXiv:1902.03383}, 2019.

\bibitem{pywren}
E.~Jonas, Q.~Pu, S.~Venkataraman, I.~Stoica, and B.~Recht, ``Occupy the cloud:
  distributed computing for the 99\%,'' in {\em Proceedings of the 2017
  Symposium on Cloud Computing}, pp.~445--451, ACM, 2017.

\bibitem{local_codes}
V.~Gupta, D.~Carrano, Y.~Yang, V.~Shankar, T.~Courtade, and K.~Ramchandran,
  ``Serverless straggler mitigation using local error-correcting codes,'' {\em
  IEEE International Conference on Distributed Computing Systems}, 2020.

\bibitem{mert_polar_serverless}
B.~Bartan and M.~Pilanci, ``Straggler resilient serverless computing based on
  polar codes,'' in {\em 2019 57th Annual Allerton Conference on Communication,
  Control, and Computing (Allerton)}, pp.~276--283, IEEE, 2019.

\bibitem{giant_nips}
S.~Wang, F.~Roosta-Khorasani, P.~Xu, and M.~W. Mahoney, ``Giant: Globally
  improved approximate newton method for distributed optimization,'' in {\em
  Advances in Neural Information Processing Systems 31}, pp.~2332--2342, Curran
  Associates, Inc., 2018.

\bibitem{boyd}
S.~Boyd and L.~Vandenberghe, {\em Convex Optimization}.
\newblock New York, NY, USA: Cambridge University Press, 2004.

\bibitem{stich2018local}
S.~U. Stich, ``{Local SGD converges fast and communicates little},'' {\em arXiv
  preprint arXiv:1805.09767}, 2018.

\bibitem{libsvm}
C.-C. Chang and C.-J. Lin, ``{LIBSVM: a library for support vector machines},''
  {\em ACM transactions on intelligent systems and technology (TIST)}, vol.~2,
  no.~3, p.~27, 2011.

\bibitem{bfgs_fletcher2013practical}
R.~Fletcher, {\em Practical methods of optimization}.
\newblock John Wiley \& Sons, 2013.

\bibitem{MLSYS2019_c8ffe9a5}
J.~Wang and G.~Joshi, ``Adaptive communication strategies to achieve the best
  error-runtime trade-off in local-update sgd,'' in {\em Proceedings of Machine
  Learning and Systems} (A.~Talwalkar, V.~Smith, and M.~Zaharia, eds.), vol.~1,
  pp.~212--229, 2019.

\bibitem{haddadpour2019local}
F.~Haddadpour, M.~M. Kamani, M.~Mahdavi, and V.~Cadambe, ``{Local SGD with
  periodic averaging: Tighter analysis and adaptive synchronization},'' in {\em
  Advances in Neural Information Processing Systems}, pp.~11080--11092, 2019.

\bibitem{patel_aymeric_nips2019}
A.~Dieuleveut and K.~K. Patel, ``Communication trade-offs for local-sgd with
  large step size,'' in {\em Advances in Neural Information Processing
  Systems}, pp.~13579--13590, 2019.

\bibitem{local_sgd_dnn}
T.~Lin, S.~U. Stich, K.~K. Patel, and M.~Jaggi, ``{Don't Use Large
  Mini-batches, Use Local SGD},'' in {\em International Conference on Learning
  Representations}, 2020.

\bibitem{swap}
V.~Gupta, S.~A. Serrano, and D.~DeCoste, ``Stochastic weight averaging in
  parallel: Large-batch training that generalizes well,'' in {\em International
  Conference on Learning Representations}, 2020.

\bibitem{jadbabaie2009distributed}
A.~Jadbabaie, A.~Ozdaglar, and M.~Zargham, ``A distributed newton method for
  network optimization,'' in {\em Proceedings of the 48h IEEE Conference on
  Decision and Control (CDC) held jointly with 2009 28th Chinese Control
  Conference}, pp.~2736--2741, IEEE, 2009.

\bibitem{tutunov2019distributed}
R.~Tutunov, H.~Bou-Ammar, and A.~Jadbabaie, ``Distributed newton method for
  large-scale consensus optimization,'' {\em IEEE Transactions on Automatic
  Control}, vol.~64, no.~10, pp.~3983--3994, 2019.

\bibitem{derezinski2018batch}
M.~Derezi\'nski, D.~Mahajan, S.~S. Keerthi, S.~Vishwanathan, and M.~Weimer,
  ``Batch-expansion training: an efficient optimization framework,'' in {\em
  International Conference on Artificial Intelligence and Statistics},
  pp.~736--744, PMLR, 2018.

\bibitem{dane}
O.~Shamir, N.~Srebro, and T.~Zhang, ``Communication-efficient distributed
  optimization using an approximate {N}ewton-type method,'' in {\em Proceedings
  of the 31st International Conference on International Conference on Machine
  Learning - Volume 32}, ICML'14, 2014.

\bibitem{disco}
Y.~Zhang and X.~Lin, ``{DiSCO: Distributed Optimization for Self-Concordant
  Empirical Loss},'' in {\em Proceedings of the 32nd International Conference
  on Machine Learning} (F.~Bach and D.~Blei, eds.), vol.~37, (Lille, France),
  pp.~362--370, PMLR, 07--09 Jul 2015.

\bibitem{cocoa}
V.~Smith, S.~Forte, C.~Ma, M.~Tak{\'a}c, M.~I. Jordan, and M.~Jaggi, ``Cocoa: A
  general framework for communication-efficient distributed optimization,''
  {\em arXiv preprint arXiv:1611.02189}, 2016.

\bibitem{aide}
S.~J. Reddi, A.~Hefny, S.~Sra, B.~P\"{o}czos, and A.~Smola, ``On variance
  reduction in stochastic gradient descent and its asynchronous variants,'' in
  {\em Proceedings of the 28th International Conference on Neural Information
  Processing Systems - Volume 2}, NIPS'15, pp.~2647--2655, 2015.

\bibitem{pmlr-v80-duenner18a}
C.~Duenner, A.~Lucchi, M.~Gargiani, A.~Bian, T.~Hofmann, and M.~Jaggi, ``A
  distributed second-order algorithm you can trust,'' in {\em Proceedings of
  the 35th International Conference on Machine Learning} (J.~Dy and A.~Krause,
  eds.), vol.~80 of {\em Proceedings of Machine Learning Research},
  pp.~1358--1366, PMLR, 10--15 Jul 2018.

\bibitem{det_avg}
M.~Derezi{\'n}ski and M.~W. Mahoney, ``Distributed estimation of the inverse
  hessian by determinantal averaging,'' in {\em Advances in Neural Information
  Processing Systems 32}, pp.~11405--11415, Curran Associates, Inc., 2019.

\bibitem{derezinski2020debiasing}
M.~Derezi{\'n}ski, B.~Bartan, M.~Pilanci, and M.~W. Mahoney, ``Debiasing
  distributed second order optimization with surrogate sketching and scaled
  regularization,'' in {\em Advances in Neural Information Processing Systems},
  vol.~33, pp.~6684--6695, 2020.

\bibitem{ghosh_comrade}
A.~Ghosh, R.~K. Maity, and A.~Mazumdar, ``Distributed newton can communicate
  less and resist byzantine workers,'' {\em arXiv preprint arXiv:2006.08737},
  2020.

\bibitem{ghosh2020communication}
A.~Ghosh, R.~K. Maity, A.~Mazumdar, and K.~Ramchandran, ``Communication
  efficient distributed approximate newton method,'' in {\em 2020 IEEE
  International Symposium on Information Theory (ISIT)}, pp.~2539--2544, IEEE,
  2020.

\bibitem{fred1}
F.~{Roosta-Khorasani} and M.~W. {Mahoney}, ``{Sub-Sampled {N}ewton Methods I:
  Globally Convergent Algorithms},'' {\em arXiv e-prints}, p.~arXiv:1601.04737,
  Jan. 2016.

\bibitem{fred2}
F.~{Roosta-Khorasani} and M.~W. {Mahoney}, ``{Sub-Sampled {N}ewton Methods II:
  Local Convergence Rates},'' {\em arXiv e-prints}, p.~arXiv:1601.04738, Jan.
  2016.

\bibitem{Fred:non-convex}
P.~Xu, F.~Roosta, and M.~W. Mahoney, ``Newton-type methods for non-convex
  optimization under inexact hessian information,'' {\em Mathematical
  Programming}, vol.~184, no.~1, pp.~35--70, 2020.

\bibitem{yao2019pyhessian}
Z.~Yao, A.~Gholami, K.~Keutzer, and M.~Mahoney, ``{PyHessian: Neural Networks
  Through the Lens of the Hessian},'' {\em arXiv preprint arXiv:1912.07145},
  2019.

\bibitem{YGSKM20_adahessian_TR}
Z.~Yao, A.~Gholami, S.~Shen, K.~Keutzer, and M.~Mahoney, ``{ADAHESSIAN}: An
  adaptive second order optimizer for machine learning,'' {\em arXiv preprint
  arXiv:2006.00719}, 2020.

\bibitem{cg}
J.~R. Shewchuk, ``An introduction to the conjugate gradient method without the
  agonizing pain,'' 1994.

\bibitem{nesterov_book}
Y.~Nesterov, {\em Introductory Lectures on Convex Optimization: A Basic
  Course}.
\newblock Springer Publishing Company, Incorporated, 1~ed., 2014.

\bibitem{svrg}
R.~Johnson and T.~Zhang, ``Accelerating stochastic gradient descent using
  predictive variance reduction,'' in {\em Advances in neural information
  processing systems}, pp.~315--323, 2013.

\bibitem{mert}
M.~Pilanci and M.~J. Wainwright, ``{N}ewton sketch: A near linear-time
  optimization algorithm with linear-quadratic convergence,'' {\em SIAM Jour.
  on Opt.}, vol.~27, pp.~205--245, 2017.

\end{thebibliography}

\clearpage

\appendix
\onecolumn

\section{Auxiliary Lemmas and their Proofs}
\label{app:aux_lemmas}
Here, we prove the auxiliary lemmas that are used in the main proofs of the paper. (For completeness, we restate the lemma statements).

\begin{lemma}
\label{lemma:subsampling_guarantee}
Let $f(\cdot)$ satisfy assumptions 1-4 and $0 <\epsilon \leq 1/2$ and $\delta <1$ be fixed constants. Then, if $s \geq \frac{4B}{\kappa\epsilon^2}\log\frac{2d}{\delta}$, the local Hessian at the $k$-th worker satisfies
\begin{align}
(1-\epsilon)\kappa \preccurlyeq \nabla^2 f^k(\w) = \H^k(\w) \preccurlyeq (1 + \epsilon)M,
\end{align}
for all $\w\in\R^d ~\text{and}~ k \in [K]$ with probability (w.p.) at least $1 - \delta$.
\end{lemma}
\begin{proof}
% Let $\{\X_1, \X_2, \cdots, \X_s\}$ be sampled uniformly at random from $[\nabla^2 f_j(\w)]_{j=1}^n$ without replacement.
At the $k$-th worker which samples $\S_k$ observations from $[n]$, the following is true by Matrix Chernoff (see Theorem 2.2 in Tropp (2011))
\begin{align}
\P(\lambda_{\min}\left(\nabla^2f^k(\w)) \leq (1-\epsilon)\kappa\right) \leq \delta_1 = d\left[\frac{e^{-\epsilon}}{(1-\epsilon)^{1-\epsilon}}\right]^{s\kappa/B}, \label{chernoff:ineq1} \\
\P(\lambda_{\max}\left(\nabla^2f^k(\w)) \geq (1+\epsilon)M\right) \leq \delta_2 = d\left[\frac{e^{\epsilon}}{(1+\epsilon)^{1+\epsilon}}\right]^{sM/B}. \label{chernoff:ineq2}
\end{align} 
Now, using the inequality $\log(1-\epsilon) \leq \frac{-\epsilon}{\sqrt{1-\epsilon}}$ for $0\leq \epsilon < 1$, we get
$$\frac{e^{-\epsilon}}{(1-\epsilon)^{1-\epsilon}} \leq e^{-\epsilon + \epsilon\sqrt{1-\epsilon}}.$$
Further, utilizing the fact that $\sqrt{1-\epsilon} \leq \frac{1}{1+\epsilon/2}$, we get
$$e^{-\epsilon + \epsilon\sqrt{1-\epsilon}} \leq e^{\frac{-\epsilon^2}{1 + \epsilon/2}} \leq e^{-\epsilon^2/4}.$$ 
Hence, we have $\delta_1 \leq de^{-s\kappa\epsilon^2/4B}$. Further, using the fact that $\log(1+\epsilon) \geq \epsilon - \epsilon^2/2$, we get
$$\frac{e^{\epsilon}}{(1+\epsilon)^{1+\epsilon}} \leq e^{-\epsilon^2/2 + \epsilon^3/2} \leq e^{-\epsilon^2/4},$$
where the last inequality follows from the fact that $\epsilon\leq 1/2$. Hence, $\delta_2 \leq de^{-sM\epsilon^2/4B}.$ Thus, by union bound and subsequently using upper bounds on $\delta_1$ and $\delta_2$, we get 
\begin{align*}
\P\left[(1-\epsilon)\kappa\I \preccurlyeq \nabla^2f^k(\w) \preccurlyeq (1 + \epsilon)M\I\right] &\geq 1 - (\delta_1 + \delta_2) \\
&\geq 1 - (de^{-s\kappa\epsilon^2/4B} + de^{-sM\epsilon^2/4B}) \\
&\geq 1 - (2de^{-s\kappa\epsilon^2/4B}),
\end{align*}
where the last inequality follows from the fact that $\kappa \leq M$. Hence, the result follows by noting that
$$(1-\epsilon)\kappa\I \preccurlyeq \nabla^2f^k(\w) \preccurlyeq (1 + \epsilon)M\I ~ \text{ w. p. at least } 1-\delta,$$
and requiring that $\delta \geq 2de^{-s\kappa\epsilon^2/4B}$ (or $s \geq \frac{4B}{\kappa\epsilon^2}\log\frac{2d}{\delta}$).

\end{proof}

\begin{lemma}
\label{lemma:local_linear_convergence}
Let the function $f(\cdot)$ satisfy assumptions 1-3, and step-size $\alpha_t^k$ that solves the line-search condition in Eq. (5). Also, let $0<\epsilon \leq 1/2$ and $0<\delta<1$ be fixed constants. Moreover, let the sample size $s \geq \frac{4B}{\kappa\epsilon^2}\log\frac{2d}{\delta}$. Then, the LocalNewton update at the $k$-th worker satisfy
\begin{align*}
     f^k(\w_{t+1}^k) - f^k(\w_t^k) \leq  - \psi ||\g_t^k||^2 ~\forall~k\in[K],
\end{align*}
w.p. at least $1 - \delta$, where $\psi = \frac{\as\beta}{M(1+\epsilon)}$. 
\end{lemma}

\begin{proof}
From Lemma \ref{lemma:subsampling_guarantee}, we know that $f^k(\cdot)$ is $M(1-\epsilon)$ smooth with probability $1 - \delta$. $M$-smoothness of a function $g(\cdot)$ implies 
\begin{equation}
g(\y) - g(\x) \leq (\y-\x)^T\nabla g(\x) + \frac{M}{2}||\y-\x||^2~\forall~ \x,\y\in \R^d.
\end{equation}
Hence, 
\begin{align}
    f^k(\w_t^k - \alpha\p_t^k) - f^k(\w_t^k) \leq (-\alpha\p_t^k)^T\g^k(\w_t^k) + \frac{M(1-\epsilon)}{2}\alpha^2||\p_t^k||^2. 
\end{align}
The above inequality is satisfied for all $\alpha\in \R$.
We know that $\alpha_t^k$, the local step-size at worker $k$, satisfies the line-search constraint in Eq. (5).
Thus, for $\alpha_t^k \in (0,1]$ to exist that satisfies the line-search condition, it is enough to find $\alpha > 0$ that satisfies
\begin{align}
  -\alpha(\p_t^k)^T\H_t^k\p_t^k + \frac{ M(1-\epsilon)}{2}\alpha^2\|\p_t^k\|^2 \leq -\alpha\beta(\p_t^k)^T\H_t^k\p_t^k  , 
\end{align}
where we have used the fact that $\g_t^k = \H_t^k\p_t^k$. Thus, $\alpha$ must satisfy
\begin{align}
    \frac{M(1-\epsilon)}{2}\alpha\|\p_t^k\|^2 \leq (1-\beta)(\p_t^k)^T\H_t^k\p_t^k.
\end{align}
Using lemma \ref{lemma:subsampling_guarantee}, we know that for sufficiently large sample-size at the $k$-th worker, we get
\begin{align}
    (1 - \epsilon) \nabla^2 f(\w) \preceq \nabla^2 f^k(\w) \preceq (1 + \epsilon) \nabla^2 f(\w)
\end{align}
with probability $1 - \delta$. Also, by $\kappa$-strong convexity of $f(\cdot)$, we know that $\nabla^2 f(\w) \succeq \kappa \I$. Thus, the local line-search constraint is always satisfied for 
$$\alpha \leq \frac{2(1-\beta)\kappa(1-\epsilon)}{M(1+\epsilon)}.$$ 
% Hence, the line search condition in \eqref{local-ss} will surely terminate for some 
% $$\alpha_t^k \geq \frac{2(1-\beta)\kappa(1-\epsilon)}{M(1+\epsilon)}.$$ 
Hence, if we choose $\as \leq \frac{2(1-\beta)\kappa(1-\epsilon)}{M(1+\epsilon)}$, or $\as \leq \frac{\kappa(1-\beta)}{M}$ for $\epsilon < 1/2$, we are guaranteed to have the line-search condition from Eq. (5) satisfied with $\alpha_t^k = \as$. This is satisfied by the line search equation in  Eq. (5).
Hence, from the line-search guarantee, we get
\begin{align}
f^k(\mathbf{w}_{t+1}^k) - f^k(\mathbf{w}_{t}^k) &\leq - \as \beta (\p_t^k)^T\g_t^k\\
&= \as \beta (\g_t^k)^T(\H_t^k)^{-1}\g_t^k,\\
&\leq - \as \beta \frac{1}{M(1+\epsilon)} \|\g_t^k\|^2,
% &\geq \frac{2(1-\beta)\kappa(1-\epsilon)}{M(1+\epsilon)} \beta \frac{1}{M(1+\epsilon)} \|\g_t^k\|^2\\
% &\geq \frac{2\beta(1-\beta)\kappa(1-3\epsilon)}{M^2}  \|\g_t^k\|^2,
\end{align}
w.p. $1 - \delta$. Here, the last inequality uses the fact that $f^k(\cdot)$ is $M(1 + \epsilon)$--smooth, that is, $\H_t^k \preceq M(1+\epsilon)\I$. This proves the desired result.
\end{proof}

\section{Proof of Theorem \ref{THM:L1}}
\label{app:thm1}
The proofs for theorems in this paper use the auxiliary lemmas in Appendix \ref{app:aux_lemmas}.
\begin{proof}
The proof of the theorem is based on the following two high probability lower bounds:
\paragraph{Case 1:}
\begin{align}
\label{eqn:key-descent-lemma}
	f(\bw_t) - f(\bw_{t+1}) \geq C \|\g(\bw_t)\|^2,
\end{align}
 where $C = \frac{\as\beta(1-\epsilon)}{2M(1+\epsilon)}$ is a constant, and
 \paragraph{Case 2}
 \begin{align}
\label{eqn:key-descent-lemma-plus}
	f(\bw_t) - f(\bw_{t+1}) \geq C_1 \|\g(\bw_t)\|^2 - \frac{\eta \Gamma}{\kappa(1-\epsilon)},
\end{align}
 where $C_1$ is a constant ($>0$) and $\eta = (1+ \sqrt{2\log (\frac{1}{\delta})})\sqrt{\frac{1}{s}}\Gamma$.

 We will prove the above result shortly, but let us 
complete the proof of the theorem assuming that Eq. \eqref{eqn:key-descent-lemma} and Eq. \eqref{eqn:key-descent-lemma-plus} are true.  

\paragraph{Case 1 (using Eq. \eqref{eqn:key-descent-lemma})}
\vspace{10pt}

\noindent Invoking the $\kappa$ strong convexity of the the function
$\f$ we have
\begin{align}
     f(\bw_t) - f(\w^*) \leq \frac{1}{2\kappa} \|\g(\bw_t)\|^2,
\end{align} 
where $\bw^*$ is the unique global minimizer of the function $\f$. Combining
the last lower bound with equation~\eqref{eqn:key-descent-lemma} we obtain
\begin{align}
  f(\bw_{t+1}) - f(\bw_{t}) \leq (1 - 2\kappa C) (f(\bw_t) - f(\w^*)), 
\end{align}
with probability $1-\delta$. Also note that
\begin{align*}
	1 > 1 - 2 \kappa C = 1 - \frac{\kappa\as\beta(1-\epsilon)}{M(1+\epsilon)} > 0,
\end{align*}
where the last inequality uses the definition of $\as$ from Eq. (6). The completes the proof of Theorem~3.2. 

\paragraph{Case 2 (using Eq. \eqref{eqn:key-descent-lemma-plus})} Using the same steps as before, and using the condition of Eq. \eqref{eqn:key-descent-lemma-plus}, we obtain Theorem~3.2.

It remains to prove the  claim~\eqref{eqn:key-descent-lemma} and \eqref{eqn:key-descent-lemma-plus}.

\paragraph{Proof of the claim~\eqref{eqn:key-descent-lemma}:}
% The proof of the claim is based on the following Lemma, which provides a high probability lower bound on the function value difference $f_k(\w_t^k) - f_k(\w_{t+1}^k):$

Recall that for $L=1$, we have
\begin{align*}
	\w_{t+1}^k = \bw_t - \alpha_t^k\p_t^k, 
	\quad \text{and} \quad
\bw_{t+1} := \frac{1}{K}\su \w_{t+1}^k = \bw_t - \frac{1}{K}\su\alpha_t^k \p_t^k,
\end{align*}
where the $\p_t^k = (\H_t^k)^{-1}\g_t^k, \H_t^k = (\H^k)^{-1}(\bw_t)$ and
$\g_t^k = \g^k(\bw_t)$. Invoking the M-smoothness of the function $\f(\cdot)$
we have
\begin{align}
    f(\bw_t) - f(\bw_{t+1}) &\geq \frac{-M}{2K^2}\|\bw_{t} - \bw_{t+1}\|^2 + \langle \g(\bw_t), \bw_{t} - \bw_{t+1}\rangle \nn\\
    &\geq \frac{-M}{2K^2}\left\|\su (\alpha_t^k)\p_t^k\right\|^2 + \langle \g(\bw_t), \frac{1}{K}\su \alpha_t^k\p_t^k\rangle \nn\\
    & \stackrel{(i)}{\geq} \frac{-M}{2K}\su (\alpha_t^k)^2\|\p_t^k\|^2 + \langle \g(\bw_t), \frac{1}{K}\su \alpha_t^k\p_t^k\rangle \nn\\
    &= \frac{1}{K}\su \left(\alpha_t^k(\p_t^k)^T\g(\bw_t) - \frac{M}{2}(\alpha_t^k)^2\|\p_t^k\|^2\right) 
    % & \stackrel{(ii)}{\geq} \frac{1}{K}\su \left(\alpha_t^k(\p_t^k)^T\g(\bw_t) - \frac{M(\alpha_t^k)^2}{2\kappa^2(1-\epsilon)^2}\|\g_t^k\|^2\right)
    \label{ineq:f(w_t) - f(w_t+1)}
\end{align}
where the inequality (i) uses the following fact
\begin{equation}\label{ineq:fact_vectors}
\left\|\frac{1}{K}\su\a^k\right\|^2 \leq \frac{1}{K}\su\|\a^k\|^2, 
\end{equation} 
for all vectors $\a^1, \a^2, \cdots, \a^K\in \R^d$.
% For inequality (ii), observe that

We now complete 
the proof by using the following bound on the first term in Eq. \eqref{ineq:f(w_t) - f(w_t+1)}. In particular, In the first case, we show that, for all $k\in [K]$ provided
\begin{align*}
    s \gtrsim \left(  \frac{\Gamma^2}{\epsilon_1^2 G^2} \log (d/\delta) \right),
\end{align*}
and $\|\g^k(\bw_t)\|\geq G$, where $\epsilon_1 >0$ (small number), we have
%
% \begin{subequations}
\begin{align}
	\alpha_t^k(\p_t^k)^T\g(\bw_t) &\geq  
	\left(\psi - \frac{\epsilon_1}{\kappa(1-\epsilon)}\right) \|\g_t^k\|^2 + \frac{\kappa(1-\epsilon)(\alpha_t^k)^2}{2}\|\p_t^k\|^2
\label{eqn:term1}    
\end{align}	
with probability at least $1- 4\delta$.

Let us substitute Eq. \eqref{eqn:term1} in equation~\eqref{ineq:f(w_t) - f(w_t+1)}, we get
\begin{align}
\label{ineq3:f(bt) - f(bt+1)}
   f(\bw_t) - f(\bw_{t+1}) &\geq 
   \frac{1}{K}\su\left[\left(\psi - \frac{\epsilon_1}{\kappa(1-\epsilon)}\right) \|\g_t^k\|^2 - \frac{(M -\kappa(1-\epsilon))(\alpha_t^k)^2}{2}\|\p_t^k\|^2\right] \nn \\
   &\geq \frac{1}{K}\su\left[\left(\psi - \frac{\epsilon_1}{\kappa(1-\epsilon)}\right) \|\g_t^k\|^2 - \frac{(M -\kappa(1-\epsilon))(\alpha_t^k)^2}{2\kappa^2(1-\epsilon)^2}\|\g_t^k\|^2\right]
\end{align}
where the last inequality follows from the fact that the function
$\f^k$ is $\kappa(1-\epsilon)$ strongly convex with probability $1 - \delta$, and thus
\begin{align}\label{ineq:norm_ptk}
    \|\p_t^k\|^2 := \|(\H_t^k)^{-1}\g_t^k\|_2^2
     \leq \|(\H_t^k)^{-1}\|_{2}^2 \|\g_t^k\|^2 
     \leq \frac{1}{\kappa^2(1-\epsilon)^2}\|\g_t^k\|^2.
\end{align}
with probability $1 - \delta$.
Now, using the upper bound on $\alpha_t^k$,
we have
\begin{align}
   f(\bw_t) - f(\bw_{t+1}) &\geq  \frac{1}{K}\su \left[\left(\psi - \frac{\epsilon_1}{\kappa(1-\epsilon)}\right)\|\g_t^k\|^2 - \frac{(M - \kappa(1-\epsilon)^2)}{2}\frac{{\as}^2}{\kappa^2(1-\epsilon)^2}\|\g_t^k\|^2\right]\nn \\
   &= \left(\psi - \frac{\epsilon_1}{\kappa(1-\epsilon)} - \frac{(M - \kappa(1-\epsilon)^2)}{2}\frac{{\as}^2}{\kappa^2(1-\epsilon)^2}\right)\frac{1}{K}\su\|\g_t^k\|^2 \nn \\
   &\geq C\frac{1}{K}\su\|\g_t^k\|^2,
\end{align}
with probability exceeding $1-6\delta$, where $C = \frac{(1-\epsilon)\psi}{2} - \frac{\epsilon_1}{\kappa(1-\epsilon)}$, and the last bound follows by substituting the value of $\alpha^*$ from equation~(6) and using the fact that 
$0 < \epsilon < 1/2$. 
Moreover, using Eq. \eqref{ineq:fact_vectors}, we get 
$$\|\g(\cdot)\|^2 \leq \frac{1}{K}\su\|\g^k(\cdot)\|^2,$$
which prove Eq. \eqref{eqn:key-descent-lemma}. 

It now remains to prove bound~\eqref{eqn:term1}.

\paragraph{Proof of bound~\eqref{eqn:term1}:}
From the uniform subsampling property (similar to Lemma \ref{lemma:subsampling_guarantee},  see Appendix~\ref{app:wo_error_floor}), we get 
\begin{align}
\label{ineq:ptk_gt_sketch}
    |(\p_t^k)^T\g(\bw_t) - (\p_t^k)^T\g^k(\bw_t)| \leq  \epsilon_1\|(\p_t^k)\|\|\g^k(\bw_t)\| \text{ w.p. } 1-\delta.
\end{align}
Thus,
\begin{align}
\label{ineq:ptk_gt1}
% (\p_t^k)^T\g(\w_t) \geq  \frac{1}{(1+\epsilon)}(\p_t^k)^T\g^k(\w_t) \geq (1-\epsilon)(\p_t^k)^T\g^k(\w_t),
(\p_t^k)^T\g(\bw_t) \geq  (\p_t^k)^T\g^k(\bw_t) - \epsilon_1\|(\p_t^k)\|\|\g^k(\bw_t)\|
\end{align}
w.p. $1-\delta$.
% , where the absolute $|\cdot|$ is removed since $(\p_t^k)^T\g^k(\w_t) > 0$ by the line-search condition and Lemma \ref{lemma:local_linear_convergence}. 
Now, since the function
$\f^k$ is $\kappa(1-\epsilon)$ strongly-convexity with probability $1 - \delta$, we have the following 
bound w.p. at least $1-\delta$:
\begin{align}
\label{ineq:ptk_gtk1}
     \alpha_t^k(\p_t^k)^T\g_t^k &\geq (f^k(\bw_t) - f^k(\w_{t+1}^k)) + \frac{\kappa(1-\epsilon)}{2}(\alpha_t^k)^2\|\p_t^k\|^2
\end{align}
Combing the equations~\eqref{ineq:ptk_gt1}-\eqref{ineq:ptk_gtk1} and using Lemma \ref{lemma:local_linear_convergence}
we have 
\begin{align*}
	\alpha_t^k(\p_t^k)^T\g(\bw_t) & \geq (f^k(\bw_t) - f^k(\w_{t+1}^k)) + \frac{\kappa(1-\epsilon)}{2}(\alpha_t^k)^2\|\p_t^k\|^2 - \epsilon_1\|(\p_t^k)\|\|\g^k(\bw_t)\|
	\\
	 &\stackrel{(i)}\geq    \psi\|\g_t^k\|^2 + \frac{\kappa(1-\epsilon)}{2}(\alpha_t^k)^2\|\p_t^k\|^2 - \epsilon_1\|(\p_t^k)\|\|\g^k(\bw_t)\| \\
	 &\stackrel{(ii)}\geq    \psi\|\g_t^k\|^2 + \frac{\kappa(1-\epsilon)(\alpha_t^k)^2}{2}\|\p_t^k\|^2 - \frac{\epsilon_1}{\kappa(1-\epsilon)}\|\g^k(\bw_t)\|^2\\
	 &= \left(\psi - \frac{\epsilon_1}{\kappa(1-\epsilon)}\right) \|\g_t^k\|^2 + \frac{\kappa(1-\epsilon)(\alpha_t^k)^2}{2}\|\p_t^k\|^2
\end{align*}
with probability exceeding $1-4\delta$, where the inequality (i) follows from Lemma~\ref{lemma:local_linear_convergence} and inequality (ii) follows from \eqref{ineq:norm_ptk}.

Note that the bound in \eqref{eqn:term1} hold for all $k \in [K]$ with probability $1-\delta_1$ (thus, the sample size increases by a factor of $K$ in the $\log(\cdot)$ term). This concludes the Case 1 of our proof. We now move to Case 2.

\paragraph{Proof of the claim~\eqref{eqn:key-descent-lemma-plus}:}
We now continue with the same analysis and show the following
\begin{align}
\label{ineq2:f(bt) - f(bt+1)}
   f(\bw_t) - f(\bw_{t+1}) \geq C_1 \frac{1}{K}\su\|\g_t^k\|^2 - \frac{\eta \Gamma}{\kappa(1-\epsilon)},
\end{align}
with probability at least $1-4\delta$.

In this case, we show that the requirement of a lower bound on $\|g^k(\bw_t)\|$ and $s$ can be relaxed at the expense of getting hit by an error floor. In particular, we show that 
\begin{align}
	\alpha_t^k(\p_t^k)^T\g(\bw_t) &\geq  
	\psi \|\g_t^k\|^2 + \frac{\kappa(1-\epsilon)(\alpha_t^k)^2}{2}\|\p_t^k\|^2 - \frac{\eta \Gamma}{\kappa(1-\epsilon)}
\label{eqn:term1plus}    
\end{align}
with probability at least $1-4\delta$, where $\eta = (1+ \sqrt{2\log (\frac{1}{\delta})})\sqrt{\frac{1}{s}}\Gamma$. Substituting this yields the bound of Eq. \eqref{ineq2:f(bt) - f(bt+1)}.

\paragraph{Proof of bound Eq. \eqref{eqn:term1plus}}: From the uniform subsampling property (see Appendix~\ref{app:with_error_floor}), we get 
\begin{align}
\label{ineq:ptk_gt_sketch_one}
    |(\p_t^k)^T\g(\bw_t) - (\p_t^k)^T\g^k(\bw_t)| \leq  \eta \|(\p_t^k)\| \text{ w.p. } 1-\delta.
\end{align}
where $\eta = (1+ \sqrt{2\log (\frac{1}{\delta})})\sqrt{\frac{1}{s}}\Gamma$. Thus,
\begin{align}
\label{ineq:ptk_gt}
% (\p_t^k)^T\g(\w_t) \geq  \frac{1}{(1+\epsilon)}(\p_t^k)^T\g^k(\w_t) \geq (1-\epsilon)(\p_t^k)^T\g^k(\w_t),
(\p_t^k)^T\g(\bw_t) \geq  (\p_t^k)^T\g^k(\bw_t) - \eta\|(\p_t^k)\|
\end{align}
w.p. $1-\delta$.
% , where the absolute $|\cdot|$ is removed since $(\p_t^k)^T\g^k(\w_t) > 0$ by the line-search condition and Lemma \ref{lemma:local_linear_convergence}. 
Now, since the function
$\f^k$ is $\kappa(1-\epsilon)$ strongly-convexity with probability $1 - \delta$, we have the following 
bound w.p. at least $1-\delta$:
\begin{align}
\label{ineq:ptk_gtk}
     \alpha_t^k(\p_t^k)^T\g_t^k &\geq (f^k(\bw_t) - f^k(\w_{t+1}^k)) + \frac{\kappa(1-\epsilon)}{2}(\alpha_t^k)^2\|\p_t^k\|^2
\end{align}
Combing the equations~\eqref{ineq:ptk_gt}-\eqref{ineq:ptk_gtk} and using Lemma \ref{lemma:local_linear_convergence}
we have 
\begin{align*}
	\alpha_t^k(\p_t^k)^T\g(\bw_t) & \geq (f^k(\bw_t) - f^k(\w_{t+1}^k)) + \frac{\kappa(1-\epsilon)}{2}(\alpha_t^k)^2\|\p_t^k\|^2 - \eta \|(\p_t^k)\|
	\\
	 &\stackrel{(i)}\geq    \psi\|\g_t^k\|^2 + \frac{\kappa(1-\epsilon)}{2}(\alpha_t^k)^2\|\p_t^k\|^2 - \eta \|(\p_t^k)\| \\
	 &\stackrel{(ii)}\geq    \psi\|\g_t^k\|^2 + \frac{\kappa(1-\epsilon)(\alpha_t^k)^2}{2}\|\p_t^k\|^2 - \frac{\eta }{\kappa(1-\epsilon)} \Gamma
\end{align*}
with probability exceeding $1-4\delta$, where the inequality (i) follows from Lemma~\ref{lemma:local_linear_convergence} and inequality (ii) follows from \eqref{ineq:norm_ptk} and the fact that $\|\g^k(\bw_t)\| \leq \Gamma$.

\end{proof}	

\section{Proof of Theorem \ref{THM:LGT1}}\label{app:thm2}
\begin{proof}
Recall from perturbed iterate analysis
\begin{align}
\bw_{t+1} = \bw_{t_0} - \sum_{\tau = t_0}^t \bp_{\tau},
\end{align}
where $\bp_\tau = \frac{1}{K}\su\alpha_\tau^k \p_\tau^k$ is the average descent direction and $\p_\tau^k = (\H_\tau^k)^{-1}\g_\tau^k$ is the local descent direction at the $k$-th worker at time $\tau$.

Similar to the proof of theorem 3.2, we next invoke the $M$-smoothness property of $f(\cdot)$ to get
\begin{align}\label{ineq:f(w_t0) - f(w_t+1)}
    f(\bw_{t_0}) - f(\bw_{t+1}) &\geq \frac{-M}{2}\|\sut\bp_\tau\|^2 + \langle \g(\bw_{t_0}), \sut \bp_\tau\rangle \nn\\
    &= \frac{-M}{2}\|\frac{1}{K}\su\sut \alpha_\tau^k\p_\tau^k\|^2 + \frac{1}{K}\su\sut\langle \g(\bw_{t_0}), \alpha_\tau^k\p_\tau^k\rangle \nn\\
    &\geq \frac{-M}{2K}\su\|\sut \alpha_\tau^k\p_\tau^k\|^2 + \frac{1}{K}\su\sut\langle \g(\bw_{t_0}), \alpha_\tau^k\p_\tau^k\rangle,
\end{align}
where the last inequality uses the fact 
\begin{equation}
\left(\su \|\a_k\|\right)^2 \leq K\su \|\a_k\|^2, ~\forall~ \a_k\in\R^d, k\in [K].
\end{equation}
Similarly, by $\kappa(1-\epsilon)$ strong-convexity of $f^k(\cdot)$, we get
\begin{align}
    f^k(\w_{t_0}^k) - f^k(\w_{t+1}^k) &\leq \frac{-\kappa(1-\epsilon)}{2}\|\sut\alpha_\tau^k\p_t^k\|^2 + \langle\g_t^k, \sut\alpha_\tau^k\p_\tau^k\rangle,
\end{align}
with probability $1-\delta$. The above inequality, when averaged across $k$, becomes
\begin{align}\label{ineq:su_ptk_gtk}
    \frac{1}{K}\su\left(f^k(\w_{t_0}^k) - f^k(\w_{t+1}^k)\right) &\leq \frac{-\kappa(1-\epsilon)}{2K}\su\|\sut\alpha_\tau^k\p_t^k\|^2 + \frac{1}{K}\su\sut\langle\g_t^k, \alpha_\tau^k\p_\tau^k\rangle
\end{align}
Moreover, similar to Eq.  \eqref{ineq:ptk_gt_sketch_one}, we get
\begin{align}
    |\r^T\g(\bw_t) - \r^T\g^k(\bw_t)| \leq  \eta \|\r\| \text{ w.p. } 1-\delta.
\end{align}
 where $\eta = (1+ \sqrt{2\log (\frac{m}{\delta})})\sqrt{\frac{1}{s}}\Gamma$.
Keeping $\r = \alpha_\tau^k\p_\tau^k$ and $\w = \bw_{t_0}$, we get 
\begin{align}\label{ineq:ptk_gt_tau}
(\alpha_\tau^k\p_\tau^k)^T\g(\bw_{t_0})  \geq (\alpha_\tau^k\p_\tau^k)^T\g^k(\bw_{t_0}) - \eta  \alpha_\tau^k\|\p_\tau^k\|,
\end{align}
w. p. $1-\delta$, where $\eta = (1+ \sqrt{2\log (\frac{m}{\delta})})\sqrt{\frac{1}{s}}\Gamma$.

Now, after combining inequalities \eqref{ineq:f(w_t0) - f(w_t+1)} and \eqref{ineq:su_ptk_gtk} using \eqref{ineq:ptk_gt_tau} to eliminate the terms  $\frac{1}{K}\su\sut\langle \g(\bw_{t_0}), \alpha_\tau^k\p_\tau^k\rangle$ and $\frac{1}{K}\su\sut\langle \g^k(\bw_{t_0}), \alpha_\tau^k\p_\tau^k\rangle$, we get
\begin{align}
   f(\bw_{t_0}) - f(\bw_{t+1}) \geq  
   \frac{1}{K}\su (f^k(\bw_{t_0}^k) - f^k(\bw_{t+1}^k) &- 
   \frac{(M - \kappa(1-\epsilon))}{2K}\su(\|\sut\alpha_\tau^k\p_t^k\|^2) \nn \\ 
   &- \frac{1}{K}\su \sum_{\tau = t_0}^t \eta \alpha_\tau^k\|\p_\tau^k\| . 
\end{align}
Also, from Lemma \ref{lemma:local_linear_convergence}, we have
\begin{equation}
f^k(\bw_{t_0}^k) - f^k(\bw_{t+1}^k) \geq \psi \sut \|\g_\tau^k\|^2.
\end{equation}
Using above, we get
\begin{align}
   f(\bw_{t_0} - f(\bw_{t+1}) \geq  
   \frac{1}{K}\psi\su\sut\|\g_\tau^k\|^2 &- 
   \frac{(M - \kappa(1-\epsilon))}{2K}\su(\|\sut\alpha_\tau^k\p_t^k\|^2) \nn\\
   &- \frac{1}{K}\su \sum_{\tau = t_0}^t \eta \alpha_\tau^k\|\p_\tau^k\|. 
\end{align}
Using triangle inequality above, we get
\begin{align}
   f(\bw_{t_0}) - f(\bw_{t+1}) \geq  
   \frac{1}{K}\psi\su\sut\|\g_\tau^k\|^2 &- 
   \frac{(M - \kappa(1-\epsilon))}{2K}\su\sut(\alpha_\tau^k)^2\|\p_\tau^k\|^2 \nn \\ 
   &- \frac{1}{K}\su \sum_{\tau = t_0}^t\eta \alpha_\tau^k\|\p_\tau^k\|. 
\end{align}
Also, since $\alpha_t^k \leq 1$ and $\|\p_\tau^k\|\leq \frac{1}{\kappa(1-\epsilon)}\|\g_\tau^k\|$, we get
\begin{align}
   f(\bw_{t_0} - f(\bw_{t+1}) &\geq  
   \frac{1}{K}\psi\su\sut\|\g_\tau^k\|^2 - 
   \frac{(M - \kappa(1-\epsilon))}{2K\kappa^2(1-\epsilon)^2}\su\sut\|\g_\tau^k\|^2 \nn\\
   &- \frac{1}{K}\su \sum_{\tau = t_0}^t \frac{\eta}{\kappa(1-\epsilon)}\|\g_\tau^k\| \nn
   \\
   &= \frac{C}{K}\su\sut \|\g_\tau^k\|^2 -  \frac{\eta L \Gamma}{\kappa(1-\epsilon)}
   % &\geq C \sut \|\g(\bw_\tau)\|^2,
\end{align}
where $C = \psi - \frac{(M - \kappa(1-\epsilon))}{2K\kappa^2(1-\epsilon)^2},$ which proves the claim. 

\end{proof}

\section{Concentration Inequalities: With and without Error Floor}
Consider a vector $v \in \mathbb{R}^d$. We have defined the following: $\g(\bw_t) = \frac{1}{n}\sum_{i}\g_i(\bw_t)$ and $\g^k(\bw_t) = \frac{1}{s}\sum_{i \in \mathcal{S}} \g_i (\bw_t)$, where $\g_i$ denotes the local gradient in worker machine $i$, and $\mathcal{S}$ is the random set consisting data points for machine $k$. Let us do the calculation in two settings:

\subsection{With error floor} 
\label{app:with_error_floor}

Here we have the error floor. Note that having an error floor is not restrictive, if we go for the adaptive variation of the algorithm, where we run GIANT for the final iterations. Since GIANT has no error floor, the final accuracy won't be affected by the error floor obtained in the first few steps of the algorithm (check if this is true).

\begin{lemma}[McDiarmid's Inequality]\label{lem:mcd}
Let $X= X_1,\ldots ,X_m$ be $m$ independent random variables taking values from some set $A$, and assume that $f: A^m \rightarrow \mathbb{R}$ satisfies the following condition (bounded differences ):
\begin{align*} 
\sup_{x_1,\ldots ,x_m,\hat{x}_i} \left|  f(x_i ,\ldots,x_i,\ldots,x_m)- f(x_i ,\ldots,\hat{x}_i,\ldots,x_m) \right| \leq c_i,
\end{align*}
for all $ i \in  \{1,\ldots ,m \}$. Then for any $\epsilon>0 $ we have 
 
 \begin{align*}
 P\left[ f(X_1,\ldots ,X_m)- \mathbb{E}[f(X_1,\ldots ,X_m)] \geq \epsilon \right] \leq \exp \left( - \frac{2\epsilon^2}{\sum_{i=1}^mc_i^2} \right).
 \end{align*}
\end{lemma}
The property described  in the following is useful  for  uniform row sampling matrix.

Let $\bS \in \mathbb{R}^{n \times s}$  be any uniform sampling matrix, then for any matrix  $\bB= [\mb_1,\ldots ,\mb_n] \in \mathbb{R}^{d\times n}$ with probability $1-\delta$ for any $\delta>0$ we have, 
\begin{align} \label{eq:uniform_sketch}
   \|\frac1n \bB\bS\bS^{\top}\mathbf{1}- \frac1n \bB\mathbf{1}\| \leq  (1+ \sqrt{2\log (\frac{1}{\delta})})\sqrt{\frac{1}{s}}\max_i \|\mb_i\|,
\end{align}
where $\mathbf{1}$ is  all ones vector.

Let us first see the justification of the above statement.The vector $ \bB\mathbf{1}$ is the sum of column of the matrix $\bB$ and  $\bB\bS\bS^{\top}\mathbf{1}$ is the sum of uniformly sampled and scaled column of the matrix $\bB$ where the scaling factor is $\frac{1}{\sqrt{sp}}$ with $p=\frac{1}{n}$. If $(i_1,\ldots ,i_s)$ is the set of sampled indices then $\bB\bS\bS^{\top}\mathbf{1}= \sum_{k\in (i_1,\ldots ,i_s) }\frac{1}{sp}\mb_k$.
 
Define the function $f(i_1,\ldots ,i_s)=\|\frac1n \bB\bS\bS^{\top}\mathbf{1}- \frac1n\bB\mathbf{1}\|$. Now consider a sampled set $(i_1,\ldots,i_{j'},\ldots ,i_s)$ with only one item (column) replaced then the bounded difference is 
\begin{align*}
    \Delta&= |f(i_1,\ldots,i_j,\ldots ,i_s)-f(i_1,\ldots,i_{j'},\ldots ,i_s)|\\
    & =|\frac1n \|\frac{1}{sp}\mb_{i_j'}-\frac{1}{sp}\mb_{i_j} \| | \leq \frac{2}{s}\max_i\|\mb_{i}\|.
\end{align*}
Now we have the expectation 
\begin{align*}
    \mathbb{E}[\| \frac1n \bB\bS\bS^{\top}\mathbf{1}- \frac1n\bB\mathbf{1}\|^2] &\leq \frac{n}{sn^2}\sum_{i=1}^n \|\mb_i\|^2= \frac{1}{s} \max_i \|\mb_i\|^2 \\
 \Rightarrow   \mathbb{E}[\|\frac1n \bB\bS\bS^{\top}\mathbf{1}- \frac1n \bB\mathbf{1}\|] & \leq \sqrt{\frac{1}{s}}\max_i \|\mb_i\|.
\end{align*}
Using McDiarmid inequality (Lemma~\ref{lem:mcd}) we have 
\begin{align*}
    P\left[\|\frac1n \bB\bS\bS^{\top}\mathbf{1}- \frac1n \bB\mathbf{1}\|\geq \sqrt{\frac{1}{s}}\max_i \|\mb_i\| + t \right] \leq \exp\left(- \frac{2t^2}{s\Delta^2} \right).
\end{align*}
Equating the probability with $\delta$ we have 
\begin{align*}
     & \exp(- \frac{2t^2}{s\Delta^2})  = \delta \\
\Rightarrow &    t =\Delta \sqrt{\frac{s}{2}\log (\frac{1}{\delta})} = \max_i \|\mb_i\|\sqrt{\frac{2}{s}\log (\frac{1}{\delta})}.
\end{align*}
Finally we have  with probability $1-\delta$
\begin{align*}
   \|\frac1n \bB\bS\bS^{\top}\mathbf{1}- \frac1n \bB\mathbf{1}\| \leq  (1+ \sqrt{2\log (\frac{1}{\delta})})\sqrt{\frac{1}{s}}\max_i \|\mb_i\|,
\end{align*}
and hence equation~\eqref{eq:uniform_sketch} is justified. 

We now apply the above in distributed gradient estimation. For the $k$-th worker machine, we have
    \begin{align*}
   \|\frac1n \bB\bS_k\bS_k^{\top}\mathbf{1}- \frac1n \bB\mathbf{1}\| \leq  (1+ \sqrt{2\log (\frac{1}{\delta})})\sqrt{\frac{1}{s}}\max_i \|\mb_i\|,
\end{align*}
with probability $1-\delta$, which implies
\begin{align*}
    \|\g^k(\bw_t) - \g(\bw_t)\| \leq (1+ \sqrt{2\log (\frac{1}{\delta})})\sqrt{\frac{1}{s}}\Gamma,
\end{align*}
with probability at least $1-\delta$ provided $\|\g_i(\bw_t)\| \leq \Gamma$ for all $i \in [m]$

Writing, $\eta = (1+ \sqrt{2\log (\frac{1}{\delta})})\sqrt{\frac{1}{s}}L$, we succinctly write
\begin{align*}
    |\langle v, \g^k(\bw_t) - \g(\bw_t) \rangle| \leq \|v\| \|\g^k(\bw_t) - \g(\bw_t)\| \leq \eta \|v\|
\end{align*}
with probability at least $1-\delta$, where $\eta = \mathcal{O}(1/\sqrt{s})$ is small.

\subsection{Without error floor}  
\label{app:wo_error_floor}
In this section, we analyze the same quantity using vector Bernstein inequality. Intuitively, we show that unless $\g(\bw_t)$ is too small, we can overcome the error floor shown in the previous calculation. In particular, we assume that
\begin{align*}
   \| \g^k(\bw_t) \| \geq G.
\end{align*}
The idea here is to use the vector Bernstein inequality. 
% \begin{lemma}
% Let $\{X_i\}_{i=1}^n$ be independent $d$ dimensional vector valued random variables, with $\mathbb{E}(X_i)=0$. Also, suppose $\|X_i\| \leq \mu$ (almost surely) and $\mathbb{E}\|X_i\|^2 \leq \sigma^2$. Then for all $t$ satisfying $0<t<\frac{\sigma^2}{\mu}$, we have
% \begin{align*}
%     \mathbb{P}\left( \| \frac{1}{n}\sum_{i=1}^n X_i \| \geq t \right) \leq d  \exp \left( - n \frac{t^2}{8\sigma^2} + 1/4 \right).
% \end{align*}
% \end{lemma}
Using the notation of Appendix~\ref{app:with_error_floor}, $\g^k(\bw_t) = \frac{1}{n} \bB \bS \bS^\top \mathbf{1}$, where $\bS$ is appropriately defined sampling matrix. Also $\g(\bw_t) = \frac{1}{n} \bB \mathbf{1}$. For the $k$-th machine, 
\begin{align*}
    \g^k(\bw_t) = \frac{1}{s}\sum_{i \in \mathcal{S}} \g_i(\bw_t),
\end{align*}
and so,
\begin{align*}
  \g^k(\bw_t) - \g(\bw_t) = \frac{1}{s}\sum_{i \in \mathcal{S}} (\g_i(\bw_t) - \g(\bw_t)),
\end{align*}
with $|\mathcal{S}|=s$. We also have $\|\g_i(\bw_t) - \g(\bw_t)\| \leq \Gamma + \Gamma =2\Gamma$, and $\mathbb{E}\|\g_i(\bw_t) - \g(\bw_t)\|^2 \leq 4\Gamma^2$. Using vector Bernstein inequality with $t = \epsilon_1 \|\g^k\|$, we obtain
\begin{align*}
    \mathbb{P}\left( \| \g^k(\bw_t) - \g(\bw_t)\| \geq \epsilon_1 \|\g^k(\bw_t)\| \right) \leq d \exp (-s \frac{\epsilon_1^2 \|\g^k\|^2}{32 \Gamma^2} + 1/4) \leq d \exp (-s \frac{\epsilon_1^2 G^2}{32L^2} + 1/4).
\end{align*}
So, as long as 
\begin{align*}
    G^2 = \Omega\left(  \frac{\Gamma^2}{\epsilon_1^2 s} \log (d/\delta)    \right),
\end{align*}
or,
\begin{align*}
    s \gtrsim \left(  \frac{\Gamma^2}{\epsilon_1^2 G^2} \log (d/\delta) \right),
\end{align*}
we have,
\begin{align*}
    |\langle v, \g^k(\bw_t) - \g(\bw_t) \rangle| \leq \|v\| \|\g^k(\bw_t) - \g(\bw_t)\| \leq \epsilon_1 \|v\| \|\g^k\|
\end{align*}
with probability at least $1-\delta$.

\section{Plots for convergence w.r.t. communication rounds}
\label{app:additional_exps}
In our main paper, we skipped the plots for convergence behavior w.r.t. communication rounds due to space constraints. In Figure \ref{fig:comm_rounds_all}, we show the convergence of adaptive \ln, GIANT, Local SGD and BFGS with communication rounds for all the five datasets considered in this paper. Again, \ln~significantly outperforms existing schemes by reducing the communication rounds by at least $60\%$ to reach the same training loss. 

\begin{figure*}[ht]
\centering
    \begin{subfigure}{.31\textwidth}
        \centering
        \includegraphics[scale=0.45]{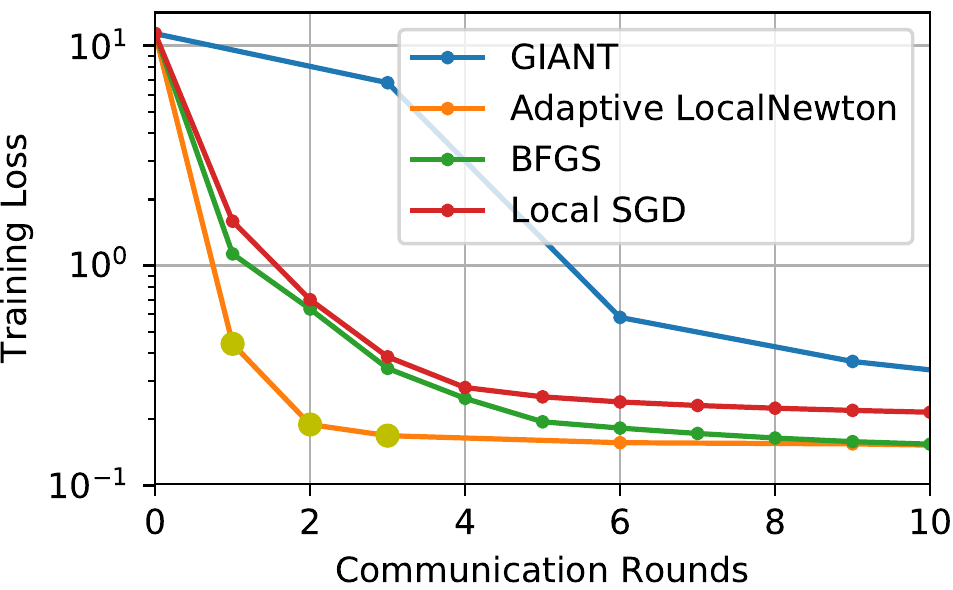}
        %\caption{50 workers}
    \end{subfigure}
    ~
    \begin{subfigure}{.31\textwidth}
      \centering
    \includegraphics[scale=0.45]{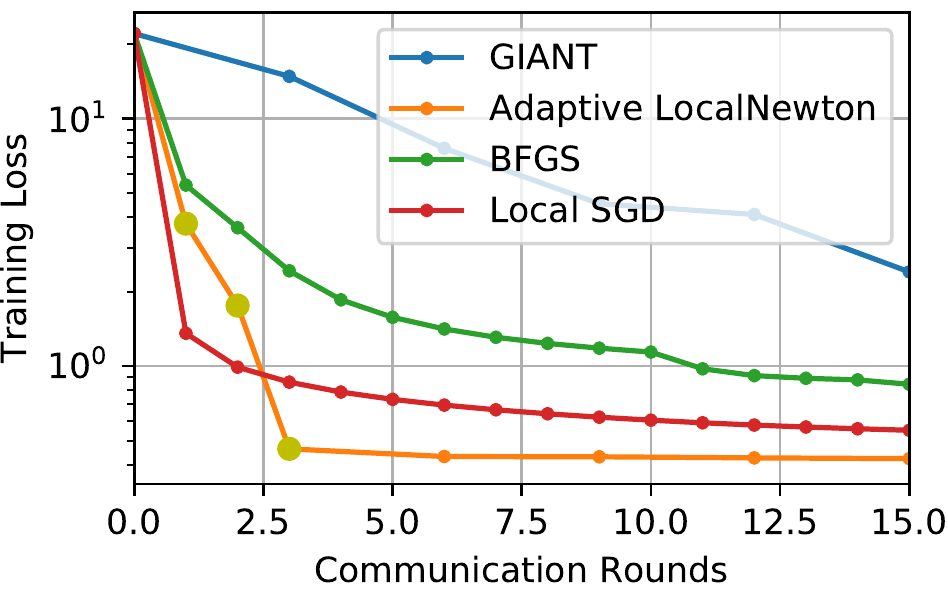}
    %\caption{200 workers}
    \end{subfigure}
~
    \begin{subfigure}{.31\textwidth}
      \centering
    \includegraphics[scale=0.45]{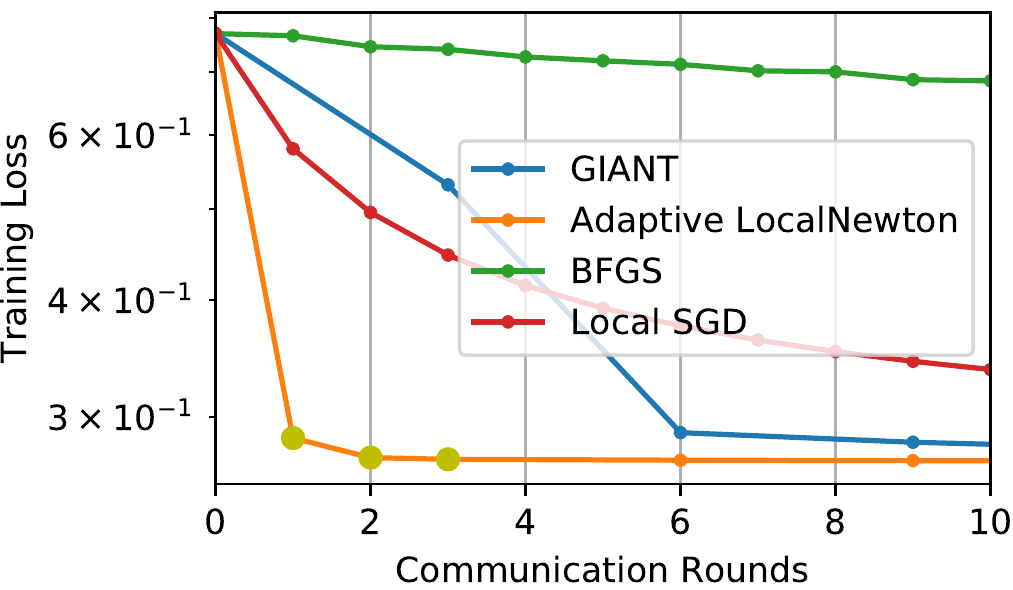}
    %\caption{200 workers}
    \end{subfigure}

    \begin{subfigure}{.31\textwidth}
        \centering
        \includegraphics[scale=0.45]{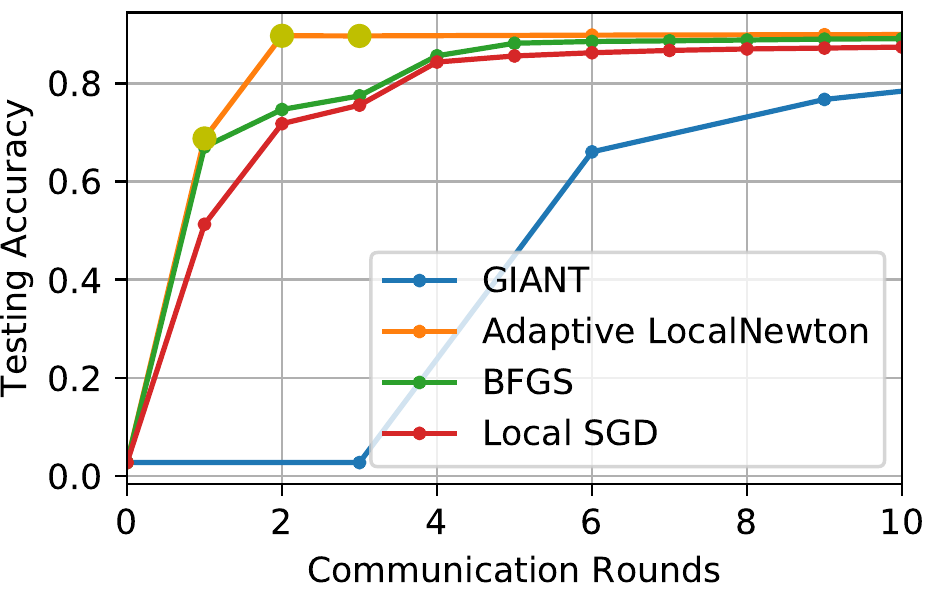}
        \caption{w8a dataset}
    \end{subfigure}
    ~
    \begin{subfigure}{.31\textwidth}
      \centering
    \includegraphics[scale=0.45]{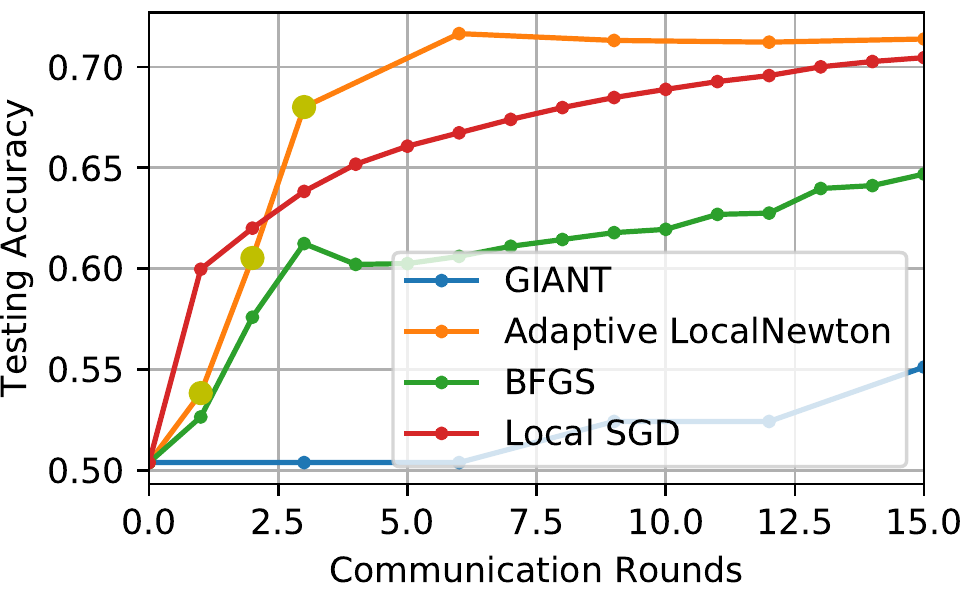}
    \caption{Covtype dataset}
    \end{subfigure}
~~~
    \begin{subfigure}{.31\textwidth}
      \centering
    \includegraphics[scale=0.45]{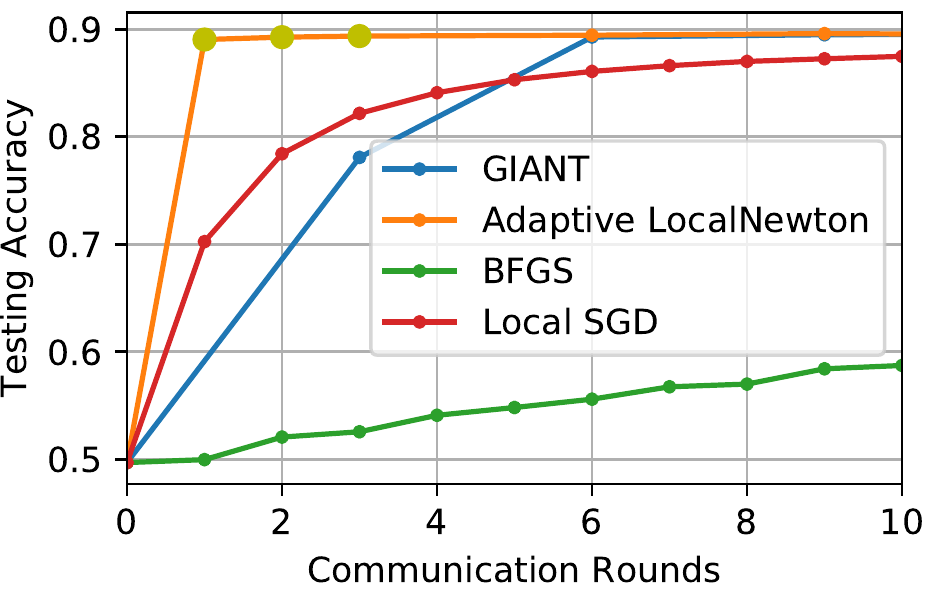}
    \caption{EPSILON dataset}
    \end{subfigure}

    \begin{subfigure}{.31\textwidth}
      \centering
    \includegraphics[scale=0.45]{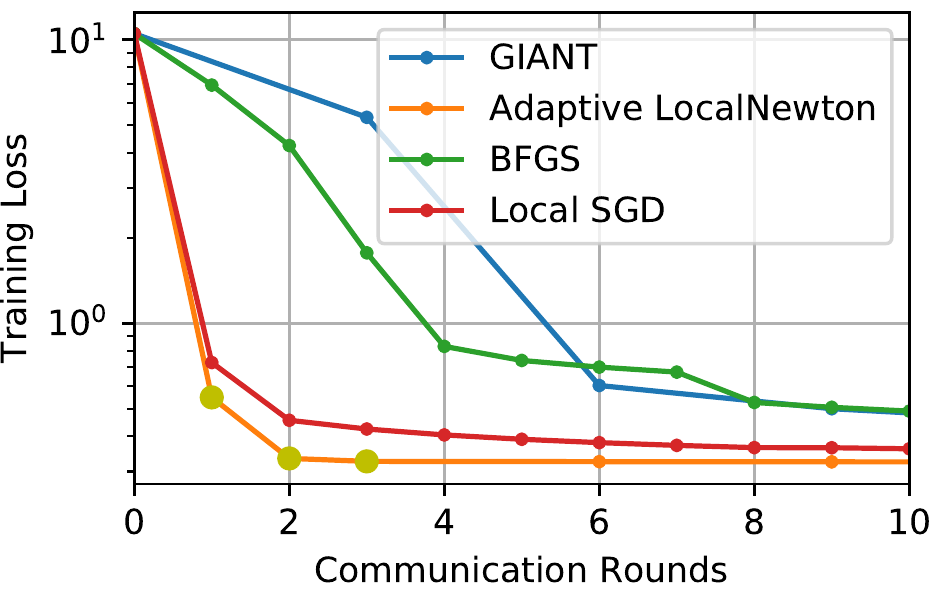}
    \end{subfigure}
~~~~~~~~~~
    \begin{subfigure}{.31\textwidth}
      \centering
    \includegraphics[scale=0.45]{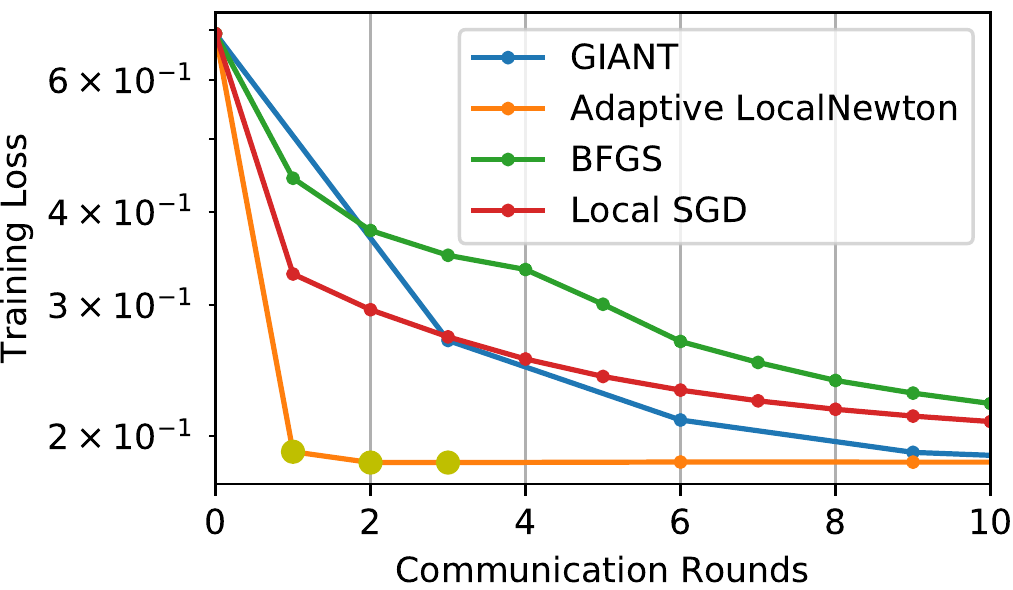}
    \end{subfigure}

    \begin{subfigure}{.31\textwidth}
      \centering
    \includegraphics[scale=0.45]{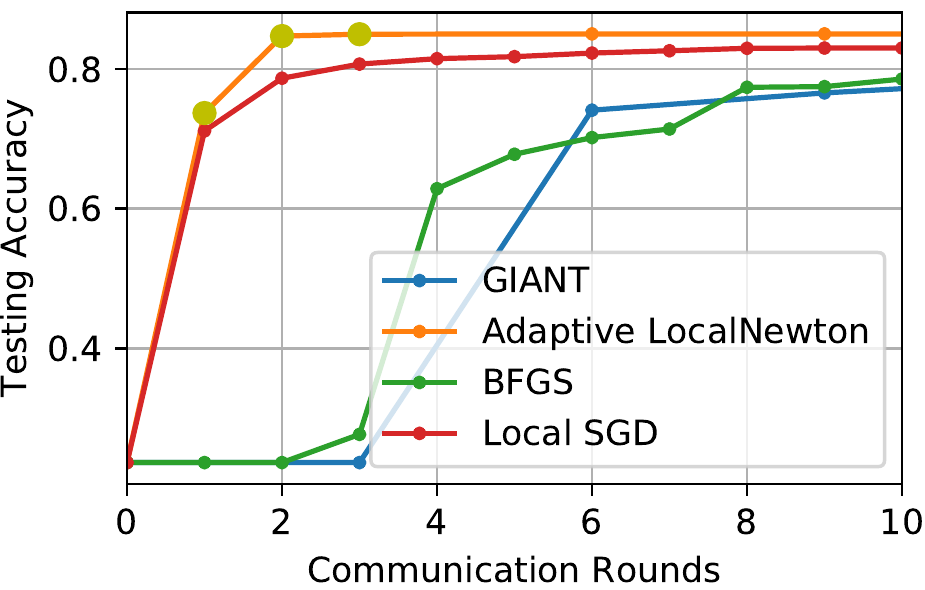}
    \caption{a9a dataset}
    \end{subfigure}
\hspace{15mm}
    \begin{subfigure}{.31\textwidth}
      \centering
    \includegraphics[scale=0.45]{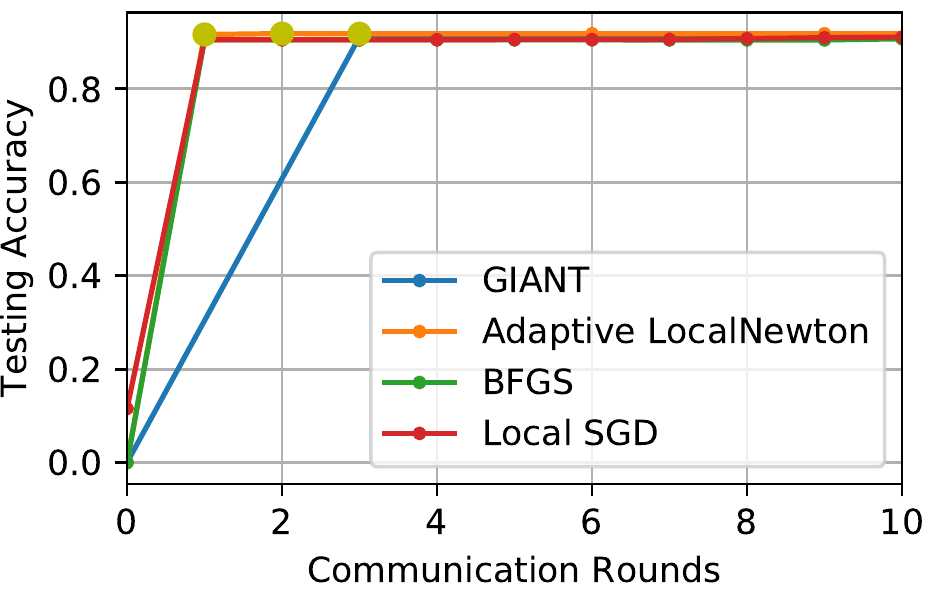}
    \caption{ijcnn1 dataset}
    \end{subfigure}
    \caption{Comparing \ln~with competing schemes w.r.t. communication rounds. Yellow dots on adaptive \ln~denote transition from larger to smaller values of $L$ (or to GIANT if $L=1$).}
    \label{fig:comm_rounds_all}
    \end{figure*}

\end{document}